\documentclass[a4paper, 12pt, oneside]{Thesis}  
\pdfoutput=1
\usepackage{amssymb}
\usepackage{url}
\usepackage{algorithm}
\usepackage{caption}
\usepackage{algpseudocode}
\usepackage[utf8]{inputenc}
\usepackage{amssymb,amsmath,amsthm}
\usepackage{epstopdf}
\usepackage{mathtools} 
\usepackage{array,multirow}
\usepackage{color}
\usepackage{array,colortbl,xcolor}
\floatname{algorithm}{Implementation}
\usepackage{titlesec}

\hypersetup{urlcolor=blue, colorlinks=true}  

\begin{document}
\frontmatter      

\title{Shortest Path Algorithms between Theory and Practice}
\authors  {\texorpdfstring
            {\href{your web site or email address}{Author Name}}
            {Author Name}
            }
\addresses  {\groupname\\\deptname\\\univname}  
\date       {\today}
\subject    {}
\keywords   {}

\maketitle

\setstretch{1.3}  

\fancyhead{}  
\rhead{\thepage}  

\pagestyle{fancy}  


\setstretch{1.3}  

\acknowledgements{
\addtocontents{toc}{\vspace{1em}}  

All praise be to Allah, the lord of the worlds, the most gracious, the most merciful.\\

\textbf{Dedication} : This thesis is dedicated to my friend \textit{Safia Salem} may Allah have mercy on her.\\ 

Foremost, I would like to express my sincere gratitude to my supervisors : Prof.Amr Elmasry and Dr.Ayman Khalafalla for guiding me through my MSc. study, introducing me to many exciting research problems, and working by my side over the past couple of years.\\

I am most indebted to my friends who supported me in many ways. I have to thank Christine Basta, Riham Salah, Rania Sherif, Sarah Ahmed, Doaa Nasser, Ahmed Elamrawy, Ahmed Ezzet, Ahmed Azaz, Ahmed Zaitoon, Ziad Ragab, Mohamed Youssef, Mostafa Gad, Walaa Wanis and Eman Ahmed  
for their good company along this journey.\\

Last but not least; I would like to thank my family..\ldots
}
\clearpage  
\Declaration{

\addtocontents{toc}{\vspace{1em}}  

I declare that no part of the work referred to in this thesis has been submitted in support of an application for another degree or qualification from this or any other University or Institution.

%
%
%
%
%
%
%
%
%
}
\clearpage  

\pagestyle{empty}  

\null\vfill
%
%
\clearpage  

\addtotoc{Abstract}  
\abstract{

Utilizing graph algorithms is a common activity in computer science. 
Algorithms that perform computations on large graphs are not always efficient. 
This work investigates the \textit{Single-Source Shortest Path (SSSP)} problem, which is considered to be one of the most important and most studied graph problems. \\

This thesis contains a review about the SSSP problem in both theory and practice. In addition, it discusses a new single-source shortest-path algorithm that achieves the same $O(n \cdot m)$ time bound as the traditional Bellman-Ford-Moore algorithm but outperforms it and other state-of-the-art algorithms. \\

The work is comprised of three parts. The first, it discusses some basic shortest-path and negative-cycle-detection algorithms in literature from the theoretical and practical point of view.
The second, it contains a discussion of a new algorithm for the single-source shortest-path problem that outperforms most state-of-the-art algorithms for several well-known families of graphs. 
The main idea behind the proposed algorithm is to select the fewest most-effective vertices to scan. We also propose a discussion of correctness, termination, and the proof of the worst-case time bound of the proposed algorithm. 
This section also suggests two different implementations for the proposed algorithm, the first runs faster while the second performs a fewer number of operations. 
Finally, an extensive computational study of the different shortest paths algorithms is conducted. The results are proposed using a new evaluation metric for shortest-path algorithms. A discussion of the outcomes and the strengths and the weaknesses of various shortest path algorithms are also included in this work.

}

\clearpage  


\pagestyle{fancy}  

\tableofcontents  

\listoffigures  

\listoftables  

\setstretch{1.5}  
\clearpage  
\listofsymbols{ll}  
{
\textbf{SSSP} & \textbf{S}ingle \textbf{S}ource \textbf{S}hortest \textbf{P}ath. \\

\textbf{ZDO} & \textbf{Z}ero \textbf{D}egree \textbf{O}nly algorithm. \\

\textbf{ZDO-Bits} & \textbf{Z}ero \textbf{D}egree \textbf{O}nly Bits algorithm. \\

\textbf{GoR} & \textbf{G}oldberg and \textbf{R}adzik algorithm. \\

\textbf{Tar} & \textbf{T}arjan algorithm. \\

\textbf{Pal} & \textbf{P}allotino algorithm. \\

\textbf{S-grids} & \textbf{S}quare grids family. \\

\textbf{W-grids} & \textbf{W}ide grids family. \\

\textbf{L-grids} & \textbf{L}ong grids family. \\

\textbf{PH-grids} & \textbf{P}ositive \textbf{H}ard grids family. \\

\textbf{NH-grids} & \textbf{N}egative \textbf{H}ard grids family. 
\\
\textbf{S-rand} & \textbf{S}parse random family. \\

\textbf{D-rand} & \textbf{D}ense random family. \\

\textbf{P-rand} & \textbf{P}otential random family. \\

\textbf{PD2S-rand} & \textbf{P}otential \textbf{D}ense to \textbf{S}parse random family. \\

\textbf{PS-rand} & \textbf{P}otential with artificial \textbf{S}ource random family. \\

\textbf{PC-rand} & \textbf{P}otential with Hamiltonian  \textbf{C}ycle random family. \\

\textbf{SPACYC} & \textbf{S}hortest \textbf{P}ath acyclic generator. \\

\textbf{FP-acyc} & \textbf{F}ully \textbf{P}ositive acyclic family. \\

\textbf{FN-acyc} & \textbf{F}ully \textbf{N}egative acyclic family. \\

\textbf{P2N-acyc} & \textbf{P}ositive to \textbf{N}egative acyclic family. \\

\textbf{SQNC} & \textbf{S}quare grids with \textbf{N}egative \textbf{C}ycles family. \\

}



\setstretch{1.3}  

\pagestyle{empty}  

\addtocontents{toc}{\vspace{2em}}  

\mainmatter	  
\pagestyle{fancy}  


\chapter{Introduction}

The single-source shortest-path problem can be defined by $(G,s,l)$, where $G= (V,A)$ is a directed weighted graph, $V$
is the set of $n$ vertices, $A$ is the set of $m$ arcs, $s$ is the source vertex, and $l : A \to \mathbb{R}$ is a length function, where $l(u,v)$ is the length of the arc $(u,v)$. The shortest path is a path of arcs with the minimum total length. The shortest path is undefined if $G$ has a cycle with negative total length. The target is to get the shortest-path tree from $s$ to all vertices in $G$ according to the length function, or to alert that $G$ has a negative cycle. 

Since Bellman \cite{bellman1958routing}, Ford \cite{ford1962flows}, and Moore \cite{moore1959shortest} have developed their $O(n \cdot m)$ shortest-path algorithm, several attempts were unsuccessful to break this worst-case bound (except for some special cases \cite{ahuja1990faster,cohen2017negative,dijkstra1959note,gabow1989faster,goldberg1995scaling,klein2010,sedgewick1986,thorup1999,thorup2000,wagner2003}). Most notable is the well-known Dijkstra algorithm that only works for graphs with non-negative arc lengths \cite{dijkstra1959note}. 
On the other hand, several heuristics were developed to outperform the Bellman-Ford-Moore algorithm in practice, including: the deque algorithm of Levit and Livshits \cite{levit1972neleneinye} and Pape \cite{pape1974implementation}, the two-queue algorithm of Pallottino \cite{pallottino1984shortest}, the topological-scan algorithm of Goldberg and Radzik \cite{goldberg1993heuristic}, and the subtree-disassembly heuristic of Tarjan \cite{Tarjan81}.

If the graph contains cycles of negative length, all the aforementioned algorithms would report it but most likely not as fast as possible.
In the literature there are several algorithms with the primary objective of promptly detecting if a negative cycle exists \cite{cherkassky2009shortest,cherkassky1999negative,goldfarb1991shortest,lawler2001combinatorial,schwiegelshohn1987shortest,Tarjan81,wong2005negative}. 

Several shortest-path algorithms are based on the general \textit{label-correcting} method \cite{Bertsekas1993,cormen1998introduction,gallo1988shortest,shier1981properties,tarjan1983data}. A \textit{potential function}, with values in $\mathbb{R}$, is defined on the set of vertices and updated throughout the algorithm. For every vertex $v$, a parent pointer $p(v)$ is defined and aims to point to the parent of $v$ forming a parent graph $G_p$. When the algorithm terminates, if $G$ has no negative cycles, $G_p$ is indeed the shortest-path tree.

For every vertex $v$, the method maintains its potential $d(v)$, parent pointer $p(v)$, and status $S(v) \in \{unreached,$ $labeled, scanned\}$. The method starts by setting $d(s)=0$ and $S(s)$ = {\em labeled};
for every other vertex: $d(v)=\infty$, $p(v)=nil$, and $S(v)=$ {\em unreached}. 
Given a potential function $d$, the reduced-cost function $l_d: A \to \mathbb{R}$ for an arc $(u,v)$ is defined
\[l_d(u,v) = l(u,v) + d(u) - d(v).\]
An arc $(u,v)$ is \textit{admissible} if it has a non-positive reduced-cost function ($l_d(u,v) \leq 0$).
The \textit{admissible graph} $G_d = (V,A_d)$ has $A$, the set of admissible arcs.
The \textit{scan} operation, defined on a labeled vertex $u$, checks all outgoing arcs from $u$ for \textit{relaxation}.
An arc $(u,v)$ is relaxed if $l_d(u,v) < 0$ by setting $d(v) \leftarrow d(u) + l(u,v)$, making $S(v)$ {\em labeled} if it is not, and setting $p(v) \leftarrow u$.
After scanning a vertex $u$, $S(u)$ becomes {\em scanned}.
The method works in rounds until no more arcs can be relaxed. For each round, the scan operation is applied to some and possibly all the labeled vertices. Different strategies for selecting labeled vertices to be scanned and their scanning order lead to different algorithms.
The method terminates if and only if $G$ does not have negative cycles. 
In this case, the parent pointers define a shortest-path tree and, for any $v\in V$, the final value for $d(v)$ is
the shortest-path distance from $s$ to $v$. 
If $G$ has negative cycles, the label-correcting method can be easily modified to find such a cycle and terminate.

This work introduces a new algorithm for the single-source shortest-path problem that runs in $O(n \cdot m)$ time. 
The proposed algorithm outperforms most state-of-the-art algorithms for several well-known families of graphs.
A description for two implementations for the algorithm is proposed, the first runs faster while the second performs fewer relaxation checks.
The thesis is organized as follows: Chapter \ref{related-work} briefly sketches some basic competitor shortest-path and negative-cycle-detection algorithms from the literature. Chapter \ref{our-algorithm} includes the proposed algorithm. Chapter \ref{our-imp} includes two implementations for the proposed algorithm and analysis on the performance. In Chapter \ref{experiments} we present a new evaluation metric for shortest-path algorithms, the experimental results, and a discussion of these outcomes. Finally, the conclusion is in Chapter \ref{conclusions}.  

\chapter{Literature Review}
\label{related-work}
There has been much work on the shortest-path problem in the literature. In this chapter, A brief description is proposed about some basic single-source shortest-path algorithms (Section \ref{SP}) and negative-cycle detection algorithms (Section \ref{NCD}).

\section{Shortest-path algorithms}
\label{SP}

\subsection{The Bellman-Ford-Moore algorithm}
The Bellman-Ford-Moore algorithm, due to Bellman \cite{bellman1958routing}, Ford \cite{ford1962flows}, and Moore \cite{moore1959shortest}, maintains the set of labeled vertices in a FIFO queue. A vertex that becomes labeled is inserted at the tail of the queue. Vertices are removed from the head of the queue to be scanned. See Figure \ref{e1}.  

\begin{theorem}
The Bellman-Ford-Moore algorithm runs in $O(n \cdot m)$ time \cite{bellman1958routing}.
\end{theorem}

\begin{figure}[!tbh]
\centering
\includegraphics[width=0.5\textwidth]{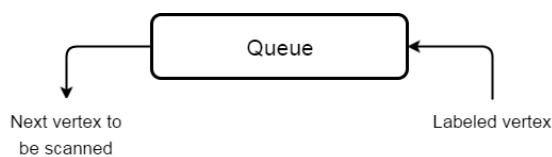}
\caption{The Bellman-Ford-Moore queue}\label{e1}
\end{figure}

\subsection{The D'Esopo-Pape algorithm}
Pape \cite{pape1974implementation} exploited a suggestion of D'Esopo \cite{pollack1960solutions} and proposed to use a deque to maintain the labeled vertices as shown in Figure \ref{e2}. A labeled vertex is inserted at the tail if it is the first time to be labeled and to the head otherwise. Vertices at the head are scanned first. 

\begin{figure}[!tbh]
\centering
\includegraphics[width=0.5\textwidth]{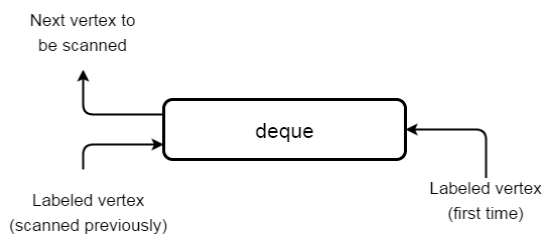}
\caption{The D'Esopo-Pape deque}\label{e2}
\end{figure}

\begin{theorem}
The D'Esopo-Pape algorithm runs in $O(n \cdot 2^n)$ time \cite{kershenbaum1981note,shier1981properties}.
\end{theorem}

\subsection{Pallottino's algorithm}
The exponential worst-case behavior for Pape's algorithm is because he uses a stack for the relabeled vertices. 
Pallottino \cite{pallottino1984shortest} suggested using a queue instead of a stack, so the data structure is composed of two connected FIFO queues $Q_1$ and $Q_2$ as shown in Figure \ref{e3}. The next vertex to be scanned is removed from the head of $Q_1$ as long as it is not empty and from the head of $Q_2$ otherwise.  A vertex that becomes labeled is inserted at the tail of $Q_1$ if it has been scanned before and at the tail of $Q_2$ otherwise. 

\begin{theorem}
Pallottino's algorithm runs in $O(n^2 \cdot m)$ time \cite{pallottino1984shortest}.
\end{theorem}

\begin{figure}[!tbh]
\centering
\includegraphics[width=1.0\textwidth]{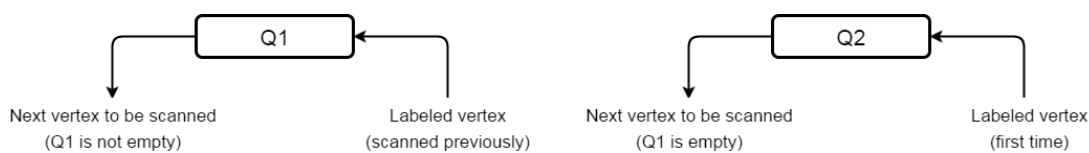}
\caption{Pallottino's queues}\label{e3}
\end{figure}

\subsection{The Goldfarb-Hao-Kai algorithm}
Goldfarb et al. \cite{goldfarb1991shortest} introduced the \textit{dynamic breadth-first search} algorithm based on maintaining levels. The algorithm maintains the label depth for each vertex; that is the number of arcs from the vertex to the root in the parent graph $G_p$. After removing a vertex from the queue, the algorithm scans it if its label depth is equal to the number of the current round. Otherwise, the vertex is put back on the queue. 

\begin{theorem}
The Goldfarb-Hao-Kai algorithm runs in $O(n \cdot m)$ time \cite{goldfarb1991shortest}.
\end{theorem}

\subsection{The Goldberg-Radzik algorithm}
Goldberg and Radzik \cite{goldberg1993heuristic} proposed the \textit{topological-scan} algorithm that achieves the same worst-case bound of the Bellman-Ford-Moore algorithm. 
Their algorithm maintains the set of labeled vertices in two queues $A$ and $B$. Each labeled vertex is in exactly one set. Initially $A=\phi$ and $B=\{s\}$.
At the beginning of each round, the algorithm uses the set $B$ to compute the set of vertices $A$ to be scanned during
this round, and resets $B$ to the empty set.  
The algorithm essentially assumes that all the vertices are unlabeled after each round.
During the round, vertices are removed according to the ordering of $A$ and scanned. 
The newly created labeled vertices (these vertices called
\textit{touched vertices}) are added to $B$. A round ends when $A$ becomes empty. The algorithm terminates when $B$ becomes empty
at the end of a round.
The algorithm computes $A$ from $B$ as follows.
(1) For every $v\in B$ that has no outgoing arc with negative reduced-cost, delete $v$ from $B$
and mark it as scanned.
(2) Let $A$ be the set of vertices reachable from $B$ in the admissible graph $G_d$. Mark all vertices in $A$ as labeled.
(3) Apply topological sorting to order $A$ so that for every pair of vertices $v$ and $w$ in $A$ 
where $(v,w) \in G_d$, $v$ precedes $w$ and will be scanned before it. 
 
\begin{theorem}
The Goldberg-Radzik algorithm runs in $O(n \cdot m)$ time \cite{goldberg1993heuristic}.
\end{theorem}

\subsection{Dijkstra-based algorithms}
 Dijkstra's algorithm \cite{dijkstra1959note} works only for graphs with non-negative arc lengths. Each round, the algorithm selects a labeled vertex with the minimum potential to be scanned next. Once a vertex is scanned, it will never be scanned again. The worst-case complexity of Dijkstra's algorithm depends on the data structure of finding the labeled vertex with the minimum potential. Suggested implementations use one-level R-heaps \cite{ahuja1990faster}, k-ary heaps \cite{cormen1998introduction}, and Fibonacci heaps\cite{fredman1987fibonacci}.

\begin{theorem}
Classical implementations for Dijkstra's algorithm run in $O(n^2)$ time or $O(m \cdot \lg n)$ time \cite{dijkstra1959note}. Using Fibonacci heaps, Dijkstra's algorithm runs in $O(m + n \cdot \lg n)$ \cite{fredman1987fibonacci}.
\end{theorem}

If there are arcs with negative lengths, some suggestions \cite{Dinitz2010,fuji1981,nakayama2013} are to handle the problem by consecutive applications of Dijkstra's method to a serious of sub-problems. Such algorithms are referred to as Dijkstra-based algorithms. The running time of the Dijkstra-based algorithms depend on the number of negative arcs and their distribution within the graph. When the number of negative arcs is large, the performance of these algorithms degrade and they cannot compete with other algorithms in practice. 

\section{Negative-cycle detection algorithms}
\label{NCD}

This section discusses some cycle-detection algorithms that can be used with the shortest-path algorithms (more details are in \cite{cherkassky1999negative}). 
\subsection{Time out}
Every labeling algorithm terminates after a certain number of labeling operations in the
absence of negative cycles. If this number is exceeded, the algorithm can stop and declare that the
graph has a negative cycle.
A major disadvantage of this method is that the number
of labeling operations used to report a negative cycle is equal to the worst-case bound.
\subsection{Admissible-graph search}
Admissible-graph search \cite{goldberg1995scaling} is based on the fact that the graph $G$ has a negative cycle if and only if 
the admissible graph $G_d$ will have a cycle.
Using depth-first search, one can periodically check if $G_d$ has a cycle in $O(n+m)$ time.
Admissible-graph search is a natural cycle-detection strategy for the Goldberg-Radzik algorithm, 
which anyhow executes a depth-first search of $G_d$ at each round to perform the topological sorting. 
\subsection{Walk to the root}
If $G$ contains a negative cycle reachable from $s$, then after a finite number of labeling operations the parent graph $G_p$ will have a cycle \cite{cherkassky1999negative}.
Suppose a relaxation operation applies to an arc $(u,v)$, 
this operation will create a cycle in $G_p$ if and only if $v$ is an ancestor of $u$ in
the current tree. Before applying the labeling operation, the algorithm follows the parent pointers
from $u$ until it reaches $v$ or $s$. If it stops at $v$, a negative cycle is found; otherwise,
the labeling operation does not create a cycle.
This method gives immediate cycle detection and can be easily combined with any
labeling algorithm. However, since paths to the root can be long, the cost of the labeling
operation becomes $O(n)$.
\subsection{Subtree traversal}
After a relaxation is applied to an arc $(u,v)$, instead of walking upwards to the root of the parent graph $G_p$ starting from $u$ 
looking for $v$, this method traverses the subtree rooted at $v$ looking for $u$ \cite{cherkassky2009shortest}.
In general, subtree traversal also increases the cost of the labeling operation to $O(n)$.
A good way to implement this  method is by using standard techniques from the network simplex method for minimum-cost flows.
\subsection{Subtree disassembly}
Subtree disassembly amortizes the subtree traversal over the work of building the subtree. 
When a relaxation is applied to an arc $(u,v)$, the subtree rooted at $v$ in $G_p$ is
traversed to find if it contains $u$ (in which case there is a negative cycle). If $u$ is not in
the subtree, all vertices of the subtree except $v$ are 
marked as \textit{unreached}. The \textit{scan} operation does not apply to these vertices until they
are labeled. Because this strategy changes some
labeled vertices to unreached, it changes the way the underlying scanning algorithm
works. A combination of the FIFO selection rule and subtree disassembly yields Tarjan's
negative-cycle detection algorithm \cite{Tarjan81}.


\chapter{The New Methodology}
\label{our-algorithm}
This chapter describes a new algorithm for the Shortest Path Problem, discusses its correctness, and proves the worst-case time bound.

\section{Main ideas}

The proposed algorithm applies the label-correcting method and as well executes in rounds. 

Recall that an arc $(u,v)$ is relaxable if it has negative reduced cost, i.e., $l_d(u,v) < 0$.
Let $A'_d$ be the set of relaxable arcs with respect to the potential function $d$.
The \textit{relaxable graph} is defined as the graph $G'_d = (V,A'_d)$.
The relaxable graph dynamically changes with every scan operation.
A vertex is declared as \textit{touched} in a round if its potential is decreased during this round 
and it has not been scanned afterward.
Initially, the source is tagged as the only touched vertex.

The objective is to use the relaxable graph to decide, with least effort, 
the most effective vertices to scan.
In contrast to the Goldberg-Radzik algorithm, the proposed algorithm considers the relaxable graph instead of the admissible graph, and only works with a subset of the vertices. Instead of generating the whole admissible graph and topologically sorting its vertices, in each round, this algorithm only looks for the touched vertices of the previous round and scans those among them having zero in-degrees within the relaxable graph of the current round. 
This work refers to the proposed algorithm as the \textit{zero-degrees-only (ZDO) algorithm}.

\begin{algorithm}[!htb]
\caption*{\bf{Algorithm} ZDO}\label{CHmain}
\begin{algorithmic} 

\State $T \gets \{s\}$
\While {$length(T) \neq 0$}
        \For {\textbf{each} vertex $v$ in T}
             \If {no arc $(-,v)$ is relaxable}
                  \State $scan(v)$
              \EndIf
				\EndFor			
         \State $T \gets$ set of touched vertices of the current iteration
\EndWhile
\end{algorithmic}
\end{algorithm}

As the next lemmas show, at least one of these vertices is scanned once and for good.

\begin{lemma}
\label{lem1}  
Consider arc $(x,y)$ on the shortest path tree. If vertex $x$ is scanned for the last time during round $r \geq 1$, then vertex $y$ will be scanned for the last time in round $r$ or round $r+1$.
\end{lemma}
\begin{proof}
After scanning vertex $x$ in round $r$, the arc $(x,y)$ is relaxed causing the potential of vertex $y$ to decrease. Vertex $y$ will then be considered for scanning during round $r$, if it is already among round-$r$ vertices, or during round $r+1$ otherwise. 
As arc $(x,y)$ is on the shortest path tree and vertex $x$ had its final potential by round $r$, then after scanning $x$ vertex $y$ must also have its final potential (unless there is a negative cycle). This implies that vertex $y$ will, by then, have zero in-degree in the relaxable graph and hence will be scanned when it is considered.  
\end{proof}

\begin{lemma}
\label{lem2}  
In the absence of negative cycles, for each round of the algorithm, at least one of the vertices is scanned for the last time during this round.
\end{lemma}
\begin{proof}
This lemma is proved by induction.  
Initially, the source is the only touched vertex and has zero in-degree before the first round,  and is indeed scanned then for the last time. 
Consider a specific round $r \geq 1$. Following the way the algorithm works,
one should have scanned all the vertices of zero in-degree in the relaxable graph of round $r$. 
By the induction hypothesis, let vertex $x$ by one of those vertices that have been scanned for the last time in round $r$. 
Since there are no negative cycles, and as long as the algorithm has not terminated, there exists an arc $(x,y)$ on the shortest path tree.
Following lemma \ref{lem1}, at least one more vertex, namely vertex $y$, is scanned once and for good before the end of round $r+1$.

\end{proof}

The next idea is to further reduce the candidates for scanning and still guarantee that at least one of the remaining vertices will not be scanned again. One can resort to the subtree disassembly idea. When an arc $(u,v)$ is relaxed, one can drop from the list of vertices to be scanned those in the subtree of $v$ in $G_p$. 
Note that the vertices of the subtree of $v$ in $G_p$ will have positive in-degrees in the relaxable graph and hence will not be scanned at this round anyhow, but dropping them from the list at this moment expedites this decision without checking the condition for each such vertex later on.
The next lemma ensures that these vertices will be touched again either way.

\begin{lemma}
Consider a vertex $v$ whose potential now drops. 
Any vertex that is in the subtree of $v$ in $G_p$ has to be touched, and hence scanned, again later. 
\end{lemma}

\begin{proof}
Consider any vertex $t$ that is in the subtree of $v$ in $G_p$. As the potential of $v$ now decreases, this indicates that the
potential of $t$ is not final, for if one relaxes all the arcs along the path from $v$ to $t$ in $G_p$ the potential of $t$ will decrease.
It follows that $t$ will be touched and scanned later.
\end{proof}

The time bound for the proposed algorithm is indicated in the next theorem.

\begin{theorem}
The proposed algorithm runs in $O(n \cdot m)$ time.
\end{theorem}

\begin{proof}
Each round of the algorithm takes $O(m)$ time to identify the candidate vertices and scan them,
as each vertex is scanned at most once per round.
During the scanning process, the algorithm disassembles the vertices rooted at each touched vertex from $G_p$. The disassembled vertices have been added before to $G_p$ as part of the scanning process. Hence, the time for the subtree disassembly is amortized over the work to scan and build $G_p$. Since at least one vertex is scanned for the last time in each round (except the sink vertex), unless there are negative cycles, the number of rounds of the algorithm is at most $n-1$. The overall running time follows.
\end{proof}

Suppose that at some round the relaxable graph is as shown in Figure \ref{graph1}. Vertices $X$, $Y$, and $Z$ have in-degree zero, while vertices $B$ and $C$ have in-degree two. Following Lemma \ref{lem2}, we are sure that at least one vertex among $X$, $Y$ or $Z$ will be scanned at this round and will never be scanned again. Intuitively, vertices $X$, $Y$ and $Z$ are the most effective vertices to be scanned in this round. Meanwhile, if we look at $G_p$ in Figure \ref{graph2} and consider a vertex $A$ whose potential now drops. Then, we will discard all vertices rooted at $A$ in $G_p$ including $X$ and $Z$ from being scanned at this round. Combining both procedures makes $Y$ the most effective vertex to scan. 

\begin{figure*}[!hpbt]
\minipage{0.4\textwidth}
  \includegraphics[width=\linewidth]{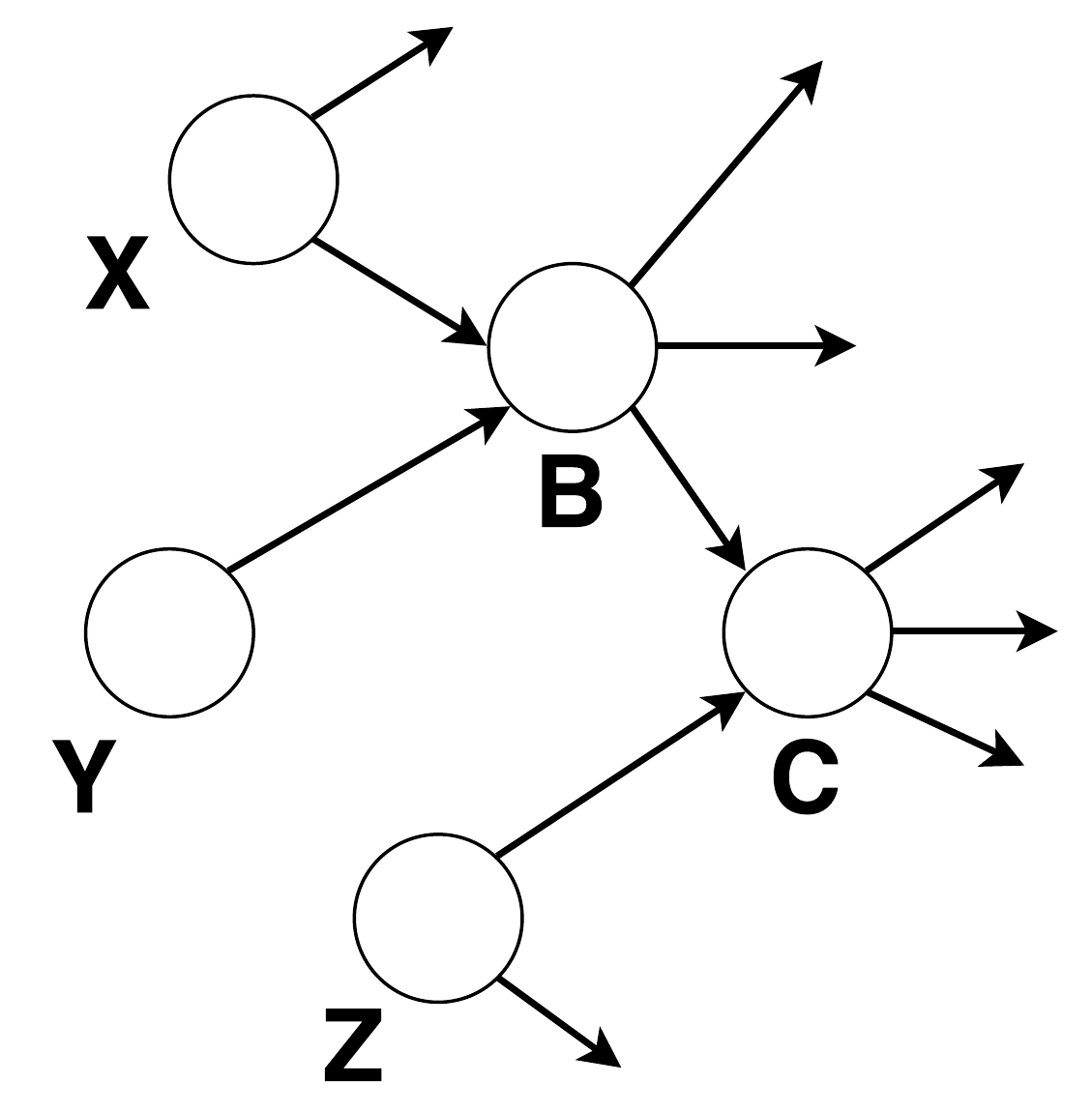}
  \caption{The relaxable graph $G'_d$}\label{graph1}
\endminipage\hfill
\minipage{0.48\textwidth}
  \includegraphics[width=\linewidth]{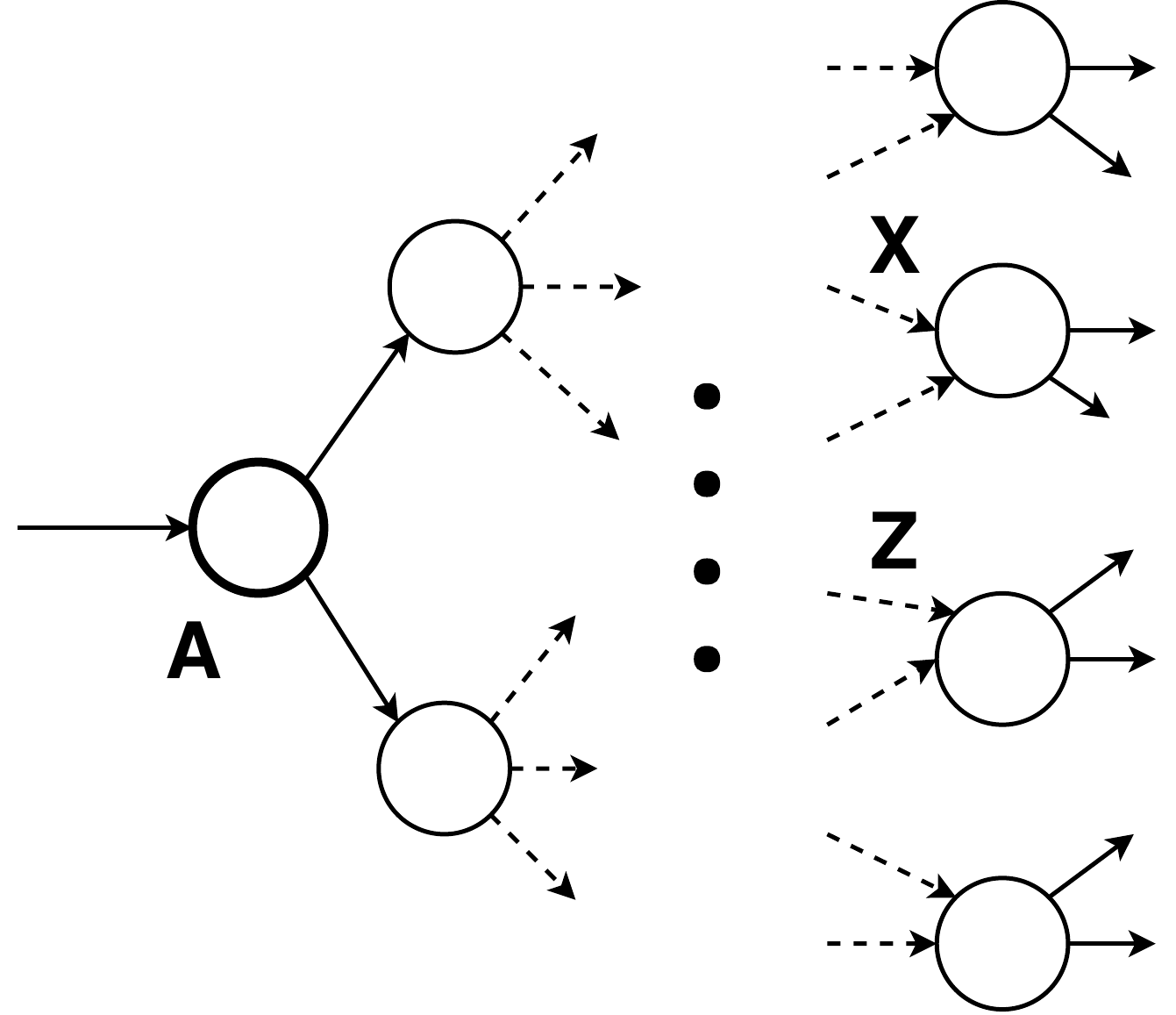}
  \caption{The current parent graph $G_p$ }\label{graph2}
\endminipage\hfill
\end{figure*}

Next chapter presents two implementations for the proposed algorithm. The two implementations differ in how to select the vertices with zero in-degree and how to scan them. The graph $G$ is presented as two arrays of adjacency lists, where \textit{v.adj\_out}[~] represents the list of outgoing arcs from vertex $v$ and \textit{v.adj\_in}[~] represents the list of incoming arcs to vertex $v$.

\chapter{Implementation}
\label{our-imp}
This chapter gives two different implementations for the proposed algorithm, the first runs faster while the second performs fewer relaxation checks.

\section{The basic implementation}

A queue $Q$ is employed to hold a set of touched vertices. 
Each vertex $v$ has an attribute $v.status$ that tells whether the vertex should not be scanned within this round even if it is in $Q$ ($v.status = inactive$), is active in $Q$ ($v.status = active$), or is physically outside $Q$ ($v.status = out$).
Initially, $Q$ contains the vertex $s$ (the only supposedly touched vertex).
The algorithm processes the vertices of $Q$ one by one, 
checks each vertex and immediately scans it if it is active and has zero in-degree in the relaxable graph $G'_d$.  
The newly touched vertices following the scan are appended to $Q$. 
To avoid duplicates, one needs to make sure that a vertex $v$ to be added to $Q$ is not already there ($v.status = out$) 
before appending it to $Q$. The algorithm terminates when $Q$ becomes empty.

To decide whether a vertex $v$ has zero in-degree in the admissible graph, the algorithm checks if any of the arcs $(u,v) \in v.adj\_in$[~] has $l_d(u,v)<0$. If it finds such an arc, the check stops realizing that $v$ does not have zero in-degree.

Within the scanning procedure, after relaxing an arc $(u,v)$, the parent graph $G_p$ is checked to disassemble the vertices of the subtree of $v$ one by one so that they would not be scanned within this round. 
This is realized by the \textit{subtree-disassembly(u,v)} procedure. For that, the vertices of $G_p$ are kept in a doubly linked list $L$ ordered in tree preorder starting from the root. The descendants of a vertex are thus consecutive in $L$.
Starting from $v$, one can traverse $L$ and the encountered vertices are removed from it and marked as inactive.
One can maintain an integer per vertex that equals its children's count in $G_p$ minus one. 
Adding these integers starting from that of $v$, one can know that the scan for $v$'s subtree is over once this sum equals $-1$.  
In the meantime, checking a negative cycle is indicated by the existence of $u$ among these vertices.

Algorithm \ref{CHalgorithm} shows the pseudo-code for this implementation (\textit{ZDO}).

\begin{algorithm}[!htb]
\caption{ZDO}\label{CHalgorithm}
\begin{algorithmic} 

\State $enqueue(Q,s)$
\State $s.status \gets active$
\While {$length(Q) \neq 0$}
        \State $v \gets  dequeue(Q)$
				
        \If {$v.status = active~\&\&~is$-$zero$-$indegree(v)$}
                  \State $scan(v)$
        \EndIf
				\State $v.status \gets out$ 
\EndWhile
\\
\Procedure{$is$-$zero$-$indegree$}{$v: vertex$}
    \For{all $u  \in  v.adj\_in[~]$}
        \If {$d(u) + l(u,v) < d(v)$} 
            \State return false   
        \EndIf
    \EndFor
    \State return true   
\EndProcedure
\end{algorithmic}
\end{algorithm}

\begin{algorithm}[!htb]
\begin{algorithmic} 
\Procedure{$scan$}{$u: vertex$}
    \For{all $v \in u.adj\_out[~]$}
        \If {$d(u) + l(u,v) < d(v)$}
            \State $d(v) \gets d(u) + l(u,v)$
						\State $p(v) \gets u$ 
            \If {$v.status = out$}
						  \State $enqueue(Q,v)$
						\EndIf	
						\State $v.status \gets active$
						\State $subtree$-$disassembly(u,v)$
        \EndIf
    \EndFor
\EndProcedure
\end{algorithmic}
\end{algorithm}


\section{Bit-vectors implementation}

One says that an arc is \textit{candidate-for-relaxation} if it may possibly have a negative reduced-cost function. 
In this implementation, each arc is accompanied with a bit to indicate if the arc is candidate-for-relaxation or not.
This implementation adopts a lazy strategy where those bits are not necessarily up to date.  
If the bit accompanying an arc $(u,v)$ is one, we still need to check the reduced-cost function to make sure if it is negative. Alternatively, if the bit is zero, we know for sure that $l_d(u,v) \geq 0$ and can safely skip checking this arc.

To store and handle those candidacy bits efficiently, we augment each vertex $v \in V$ with two bit vectors (stored in computer words) one for the incoming arcs to $v$, $v.in\_bits$[~], and one for the outgoing arcs from $v$, $v.out\_bits$[~]. 
The bit in position $j$ of \textit{$v.in\_bits$}[~] is affiliated with the incoming arc number $j$ in \textit{v.adj\_in}[~], while the bit in position $j$ of \textit{$v.out\_bits$}[~] is affiliated with the outgoing arc number $j$ in \textit{v.adj\_out}[~].
Also, each arc $(u,v)\in A$ has two indices $pos_u$ and $pos_v$. The $pos_u$ index refers to the arc position in the source-vertex bit vector $u.out\_bits$[~], while the $pos_v$ index refers to the arc position in the destination-vertex bit vector $v.in\_bits$[~].

This implementation relies on the following routines \cite{ffs} that enable us to deal with bit vectors:
  \begin{itemize}
  \item $set(pos, vec)$: set the bit at position $pos$ in bit vector $vec$ to one.
  \item $clear(pos, vec)$: reset the bit at position $pos$ in bit vector $vec$ to zero.
  \item {\it ffs}$(vec)$: {\it find-first-set} one bit in bit vector $vec$, or $\phi$ if all the bits are zeros.
  \end{itemize}

Theoretically, it is possible to implement these commands to run in worst-case constant time \cite{FredmanW93}, but this work does not use these implementations. Alternatively, many architectures include instructions to rapidly perform the find-first-set operation, and a number of compilers supply efficient built-in routines to utilize these hardware instructions \cite{ffs}.

Using bit vectors enables us to check whether a vertex has zero in-degree in the admissible graph by only checking the incoming arcs candidates for relaxation. Also, while scanning a vertex, it is useful to jump over outgoing arcs and skip those not candidates for relaxation.
When the potential of a vertex is dropped, its outgoing arcs with zero candidacy bits are checked. 
Among those, the candidacy bits for the arcs with negative reduced costs are set to one, 
and the target vertex $v$ of the arc is marked for not to be scanned within this round ($v.status \gets inactive$).  

Actually, ZDO-Bits trades the number of relaxation checks by answering the question whether $d(u) + l(u,v) < d(v)$ through setting and checking the bit words inside the vertices. So, the overhead of fetching the words from the vertices before the setting and clearing operations would result in this implementation becoming slow in practice.

Algorithm \ref{CHalgorithm2} shows the pseudo-code for this implementation (\textit{ZDO-Bits}).  

\begin{algorithm}[!htb]
\begin{algorithmic}
\caption{ZDO-Bits}\label{CHalgorithm2}
\State $enqueue(Q, s)$
\State $s.status \gets active$
\State $s.out\_bits[~] \gets {\bf \overline{1}}$
\While {$length(Q) \neq 0$}
    \State $v \gets  dequeue(Q)$
    \If {$v.status = active~\&\&~is$-$zero$-$indegree(v)$}
        \State $scan(v)$
    \EndIf
		\State $v.status \gets out$
\EndWhile
\\
\Procedure{$is$-$zero$-$indegree$}{$v: vertex$}
    \While {$v.in\_bits[~] \neq 0$}
            \State $j \gets$ {\it ffs}$(v.in\_bits[~])$
					  \State $u \gets v.adj\_in[j]$
            \If {$d(u) + l(u,v) < d(v)$}
                \State return {\it false}
            \EndIf
            \State $clear(pos_u(u,v) ,u.out\_bits[~])$
            \State $clear(j ,v.in\_bits[~])$
     \EndWhile
     \State return {\it true}
\EndProcedure
\\
\Procedure{$scan$}{$u: vertex$}
    \While {$u.out\_bits[~] \neq 0$}
        \State $j \gets$ {\it ffs}$(u.out\_bits[~])$
				\State $v \gets u.adj\_out[j]$ 
        \State $clear(j, u.out\_bits[~])$
        \State $clear(pos_v(u,v), v.in\_bits[~])$
        \If {$d(u) + l(u,v) < d(v)$}
            \State $d(v) \gets d(u) + l(u,v)$
						\State $p(v) \gets u$
						\If {$v.status = out$}
						  \State $enqueue(Q,v)$ 
						\EndIf	
						\State $v.status \gets active$
						\State $update$-$bit$-$vectors(v)$  
            \State $subtree$-$disassembly(u,v)$
        \EndIf
    \EndWhile
\EndProcedure
\end{algorithmic}
\end{algorithm}
\begin{algorithm}[!htb]
\begin{algorithmic}
\Procedure{$update$-$bit$-$vectors$}{$u: vertex$}
      \State $vec \gets u.out\_bits[~] \oplus {\bf \overline{1}}$ 
			\While {$vec \neq 0$}
            \State $j \gets$ {\it ffs}$(vec)$
					 	\State $clear(j,vec)$
						 \State $v \gets u.adj\_out[j]$ 
             \If {$d(u) + l(u,v) < d(v)$} 
                 \State $set(j, u.out\_bits[~])$
                 \State $set(pos_v(u,v), v.in\_bits[~])$
								 \State $v.status \gets inactive$
             \EndIf
        \EndWhile
\EndProcedure				
\end{algorithmic}
\end{algorithm}

\chapter{Experimental Results}
\label{experiments}
This chapter presents experimental setup, a new evaluation metric for shortest-path algorithms, results for the state-of-the-art shortest-path algorithms (including ours) on several families of graphs, and discuss those results and comment on them.
 
\section{Experimental setup}

This work refers to the Bellman-Ford-Moore algorithm with the subtree-disassembly heuristic of Tarjan as \textit{Tar}, the Pallotino algorithm with the subtree-disassembly heuristic as \textit{Pal}, and the Goldberg-Radzik algorithm with the admissible-graph search as \textit{GoR}. 
A comparison is performed between the new proposed algorithm with these state-of-the-art implementations.
This work does not consider Dijkstra-based algorithms, which perform poorly when the number of negative arcs is large.
Our algorithm, in both the ZDO and ZDO-Bits implementations, detects negative cycles using the subtree-disassembly heuristic.

Our experiments are conducted on a windows 7 machine with core i7 2GHz processor and 8GB memory. Our code is written in C and compiled with gcc compiler using O1 optimization.
The bit vectors in ZDO-Bits are implemented as 64-bits integers. This implementation depends on the standard bitwise operations in the implementation of the \textit{set()} and \textit{clear()} routines. For the \textit{ffs()} routine, gcc compiler has a fast built-in implementation called \textit{\_\_builtin\_ffs}\cite{ffs}.
The problems generators and the implementation for the algorithms (Tar, Pal, GoR) were developed by the author of \cite{goldberg1995scaling}. The problem generator for the worst case for GoR is written by us in java. 

\clearpage
\section{Evaluation metrics}
One can define the \textit{number of checks per arc} as a follows:
\begin{equation}
  \frac{total\ number\ of\ relaxation\ checks}{m}  
\end{equation}
where the total number of relaxation checks is a counter for checking that $d(u) + l(u,v) < d(v)$ for every arc $(u,v) \in A$ and $m$ is the number of arcs. This new performance metric is more realistic than the number of scans per vertex.
As a supportive argument, the new proposed algorithm does not need to check all in-coming arcs to decide whether the vertex has zero in-degree in $G'_d$. Instead, it stops once finding the first arc that fulfills the condition. 

The number of relaxation checks are split into two categories: auxiliary checks and main checks.
The auxiliary checks are the extra checks to determine which vertices are to be scanned first.
In our algorithm, the checks done to identify if a vertex has zero in-degree and those done to update the bit vectors are auxiliary checks.
The main checks are the checks done when a vertex is scanned, and each is possibly followed by an update of the potential value. 
As will be illustrated later by the experimental results, 
the main checks are more influential on the running times than the auxiliary checks.
We also plot the running time for different algorithms. The running time is the user CPU time excluding the input and output times. Following \cite{cherkassky2009shortest}, each data point (in tables and plots) represents the average over five runs with the same generator parameters except for the pseudo-random generator seed.  We did not present the data points with running times greater than $100$ seconds, as we consider the corresponding algorithms slow in these cases.
\clearpage
\section{Experimental results}
This section investigates the different families of graphs. For each family, a comparison is performed between the proposed algorithm and the \textit{state-of-the-art} algorithms; mainly: Tar, Pal, and GoR.

\subsection{Star structure/Bad-GoR}
The Star family is composed of a central vertex that has a large number of incoming and outgoing arcs. The incoming arcs are coming from a chain of vertices. All arcs in the graph have a weight of -1. Figure \ref{bestzdobits} shows the Star structure graph. Assuming that the in-degree of the central vertex is equal to its out-degree ($k$). Table \ref{t2best} shows the number of relaxation checks for different $k$ values. Both GoR and ZDO-Bits achieve the minimum number of checks compared to other algorithms. Figure \ref{fbest} presents the running time for different algorithms on this family. Pal and Tar keep scanning the central vertex $k$ times, so they lose the competition. Since the graphs are fully negative, GoR solves the  problem in one DFS pass and performs very few relaxation checks. Both ZDO and ZDO-Bits scan the central vertex once (after scanning the chain). In addition, ZDO-Bits encodes the in-degree of the central vertex using the bit vectors. So, it performs the fewest number of relaxation checks, that is two orders of magnitude faster than ZDO.

\begin{figure}[!hpbt]
\centering
\includegraphics[width=0.45\textwidth]{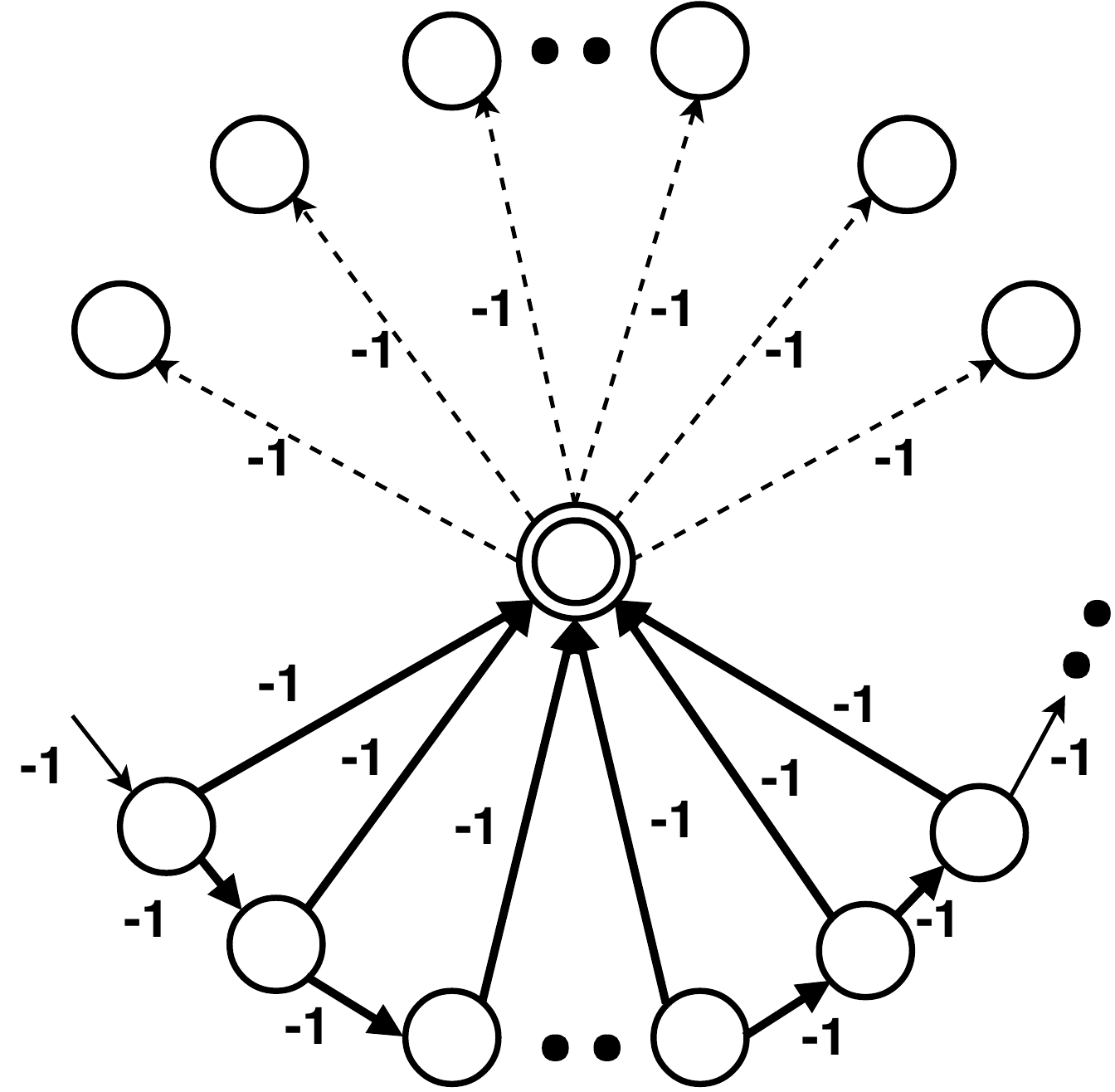}
\caption{Star structure.}\label{bestzdobits}
\end{figure}

\begin{table}[!ht]
\centering
\caption{\\ Number of relaxation checks per arc for different algorithms on Star family\label{t2best}}{%
\begin{tabular}{|c||c|c|c|c|c|c|c|c|} \hline 
\multirow{2}{*}{$K/10^3$}  & \multirow{2}{*}{Pal} & \multicolumn{2}{c|} {GoR} & \multirow{2}{*}{Tar} & \multicolumn{2}{c|}  {ZDO} & \multicolumn{2}{c|}  {ZDO-Bits} \\ \cline{3-4} \cline{6-7} \cline{8-9}
 &  & aux	& main &   &  aux	& main &  aux	& main \\ \hline \hline

10 & 3333.78 & 2	& 1 & 1667.389  & \cellcolor{lightgray} 834.361	& \cellcolor{lightgray} 1 & \cellcolor{lightgray} 1.167	& \cellcolor{lightgray} 1\\ \hline

20 & 6667.111 & 2	& 1 & 3334.056 & \cellcolor{lightgray} 1667.694	& \cellcolor{lightgray} 1 & \cellcolor{lightgray} 1.167	& \cellcolor{lightgray} 1\\ \hline

30 & 10000.444 & 2	& 1 & 5000.722 & \cellcolor{lightgray} 2501.028	& \cellcolor{lightgray} 1 & \cellcolor{lightgray} 1.167	& \cellcolor{lightgray} 1\\ \hline

40 & 13333.778 & 2	& 1 & 6667.389 & \cellcolor{lightgray} 3334.361	& \cellcolor{lightgray} 1 & \cellcolor{lightgray} 1.167	& \cellcolor{lightgray} 1\\ \hline

50 & & 2	& 1 &  8334.06 & \cellcolor{lightgray} 4167.694	& \cellcolor{lightgray} 1 & \cellcolor{lightgray} 1.167	& \cellcolor{lightgray} 1\\ \hline

60 & & 2	& 1 &  10000.72 & \cellcolor{lightgray} 5001.028	& \cellcolor{lightgray} 1 & \cellcolor{lightgray} 1.167	& \cellcolor{lightgray} 1\\ \hline

70 & & 2 & 1 &   & \cellcolor{lightgray} 5834.361	& \cellcolor{lightgray} 1 & \cellcolor{lightgray} 1.167	& \cellcolor{lightgray} 1\\ \hline

80 & & 2 & 1 &  & \cellcolor{lightgray} 6667.694	& \cellcolor{lightgray} 1  & \cellcolor{lightgray} 1.167	& \cellcolor{lightgray} 1 \\ \hline

90 & & 2 & 1 &  & \cellcolor{lightgray} 7501.028	& \cellcolor{lightgray} 1 & \cellcolor{lightgray} 1.167	& \cellcolor{lightgray} 1\\ \hline

100 & & 2 & 1 &  & \cellcolor{lightgray} 8334.361	& \cellcolor{lightgray} 1 & \cellcolor{lightgray} 1.167	& \cellcolor{lightgray} 1 \\ \hline
\end{tabular}}
\end{table}

\begin{figure}[!hpbt]
\centering
\includegraphics[width=0.9\textwidth]{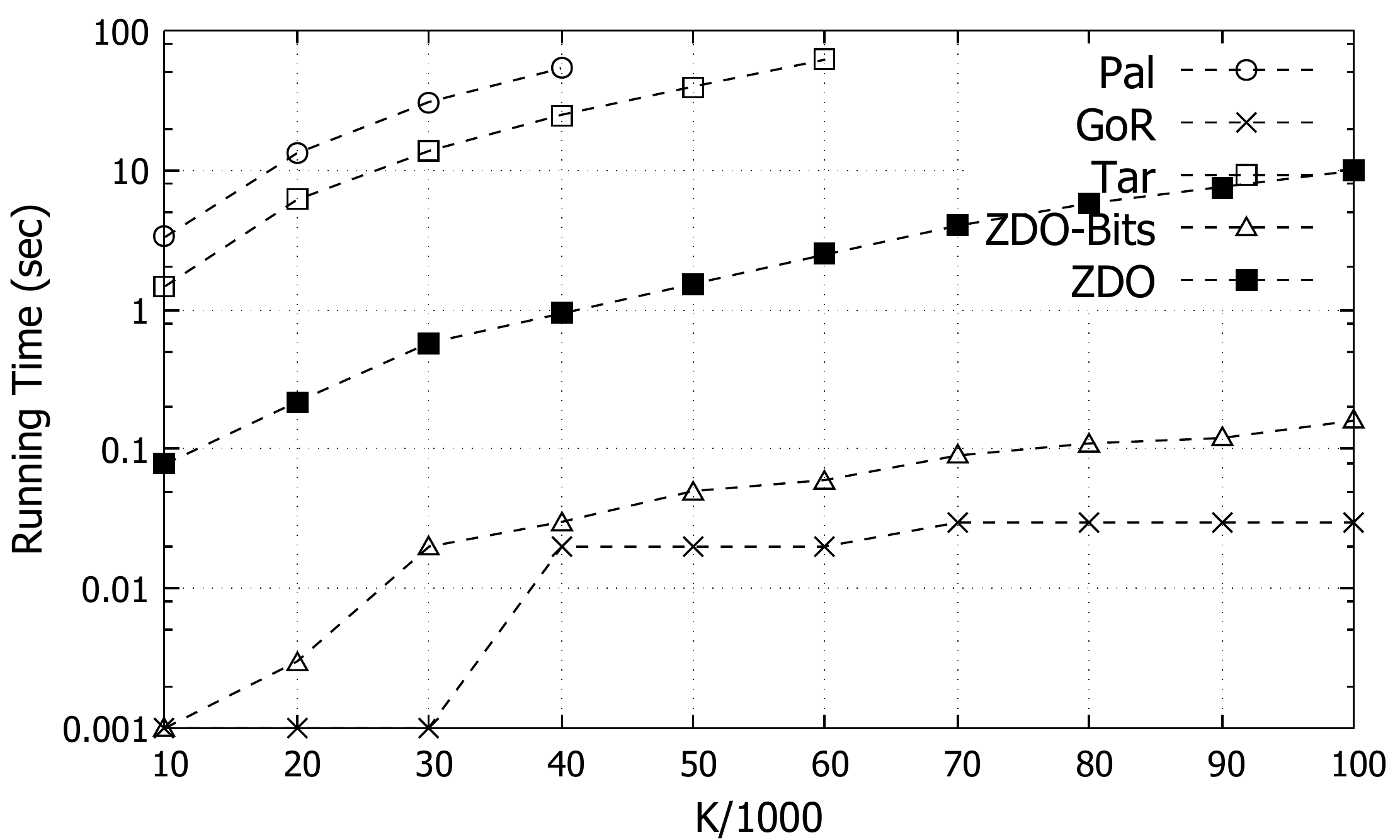}
\caption{Running time for different algorithms on Star family.}\label{fbest}
\end{figure}

\par
Now, if one plugs positive weights to the chain, the problem becomes also harder for GoR.
As per \cite{cherkassky2009shortest}, the worst-case behavior of GoR is illustrated
when applied to a family of graphs ({\it Bad-GoR}). 
For this family, the number of vertices and arcs are functions of a parameter $k$ that indicates how many times a given gadget is repeated.
The graph contains $2k+1$ vertices and $3k-1$ arcs. Vertices from 1 to $k$ are connected with a path that contains arcs $(i, i+1)$ for $1 \leq i < k$. Vertex $k+1$ has $k$ incoming arcs from the first $k$ vertices and $k$ outgoing arcs to the other $k$. Arc lengths have values as follows: $l(1,2)= -3k, l(1,k+1) = -1$, $l(i,k+1) = 2(k-i)$ for 
$2 \leq i \leq k$, $l(i,i+1) = 1$ for $2 \leq i < k$, and $l(k+1,i)=-1$ for $k+2 \leq i \leq 2k+1$. Figure \ref{f3} gives an example for $k = 7$.
GoR keeps scanning vertex $k+1$ every round until the path with vertices $i$ ($1 \leq i \leq k$) is scanned completely. As the ZDO and ZDO-Bits algorithms scan only the zero in-degree vertices in $G'_d$, they keep skipping vertex $k+1$ as it does not have zero in-degree.

\begin{figure}[!hpbt]
\centering
\includegraphics[width=0.45\textwidth]{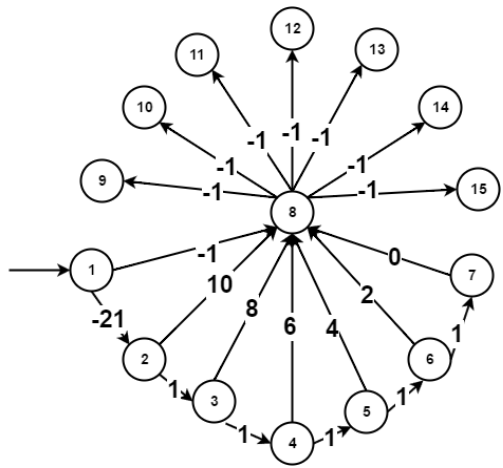}
\caption{Bad-GoR family structure ($k=7$) }\label{f3}
\end{figure}

Table \ref{t2} presents the results of Bad-GoR family with different values of $k$. 
The number of relaxation checks are improved significantly by our algorithm compared to other algorithms.
Figure \ref{f4} presents the running time for different algorithms on this family. Pal, GoR and Tar algorithms lose the competition. Their running times are greater than 100 seconds for large graph sizes, while our algorithm solves these instances more efficiently. 
The ZDO-Bits implementation is about three orders of magnitude faster than the other algorithms.
In contrast to the later figures, we use a linear scale for the abscissa axis to gradually monitor the wide variation in performance between our algorithm compared to the other algorithms.

\begin{table}[!ht]
\centering
\caption{\\ Number of relaxation checks per arc for different algorithms on Bad-GoR family\label{t2}}{%
\begin{tabular}{|c||c|c|c|c|c|c|c|c|} \hline 
\multirow{2}{*}{$K/10^3$}  & \multirow{2}{*}{Pal} & \multicolumn{2}{c|} {GoR} & \multirow{2}{*}{Tar} & \multicolumn{2}{c|}  {ZDO} & \multicolumn{2}{c|}  {ZDO-Bits} \\ \cline{3-4} \cline{6-7} \cline{8-9}
 &  & aux	& main &   &  aux	& main &  aux	& main \\ \hline \hline

10 & 3333.778 & 3335.028	& 1667.389 & 1667.389  & \cellcolor{lightgray} 834.361	& \cellcolor{lightgray} 1 & \cellcolor{lightgray} 1.167	& \cellcolor{lightgray} 1\\ \hline

20 & 6667.111 & 6668.361	& 3334.056 & 3334.056 & \cellcolor{lightgray} 1667.694	& \cellcolor{lightgray} 1 & \cellcolor{lightgray} 1.167	& \cellcolor{lightgray} 1\\ \hline

30 & 10000.444 & 10001.694	& 5000.722 & 5000.722 & \cellcolor{lightgray} 2501.028	& \cellcolor{lightgray} 1 & \cellcolor{lightgray} 1.167	& \cellcolor{lightgray} 1\\ \hline

40 & 13333.778 & 13335.028	& 6667.389 & 6667.389 & \cellcolor{lightgray} 3334.361	& \cellcolor{lightgray} 1 & \cellcolor{lightgray} 1.167	& \cellcolor{lightgray} 1\\ \hline

50 & & 16668.361	& 8334.056 & 8334.056 & \cellcolor{lightgray} 4167.694	& \cellcolor{lightgray} 1 & \cellcolor{lightgray} 1.167	& \cellcolor{lightgray} 1\\ \hline

60 & & 20001.694	& 10000.722 & 10000.722 & \cellcolor{lightgray} 5001.028	& \cellcolor{lightgray} 1 & \cellcolor{lightgray} 1.167	& \cellcolor{lightgray} 1\\ \hline

70 & & 	&  &  & \cellcolor{lightgray} 5834.361	& \cellcolor{lightgray} 1 & \cellcolor{lightgray} 1.167	& \cellcolor{lightgray} 1\\ \hline

80 &&& &  & \cellcolor{lightgray} 6667.694	& \cellcolor{lightgray} 1  & \cellcolor{lightgray} 1.167	& \cellcolor{lightgray} 1 \\ \hline

90 &&& &  & \cellcolor{lightgray} 7501.028	& \cellcolor{lightgray} 1 & \cellcolor{lightgray} 1.167	& \cellcolor{lightgray} 1\\ \hline

100 &&& &  & \cellcolor{lightgray} 8334.361	& \cellcolor{lightgray} 1 & \cellcolor{lightgray} 1.167	& \cellcolor{lightgray} 1 \\ \hline
\end{tabular}}
\end{table}
\begin{figure}[!ht]
\centerline{\includegraphics[width=0.9\textwidth]{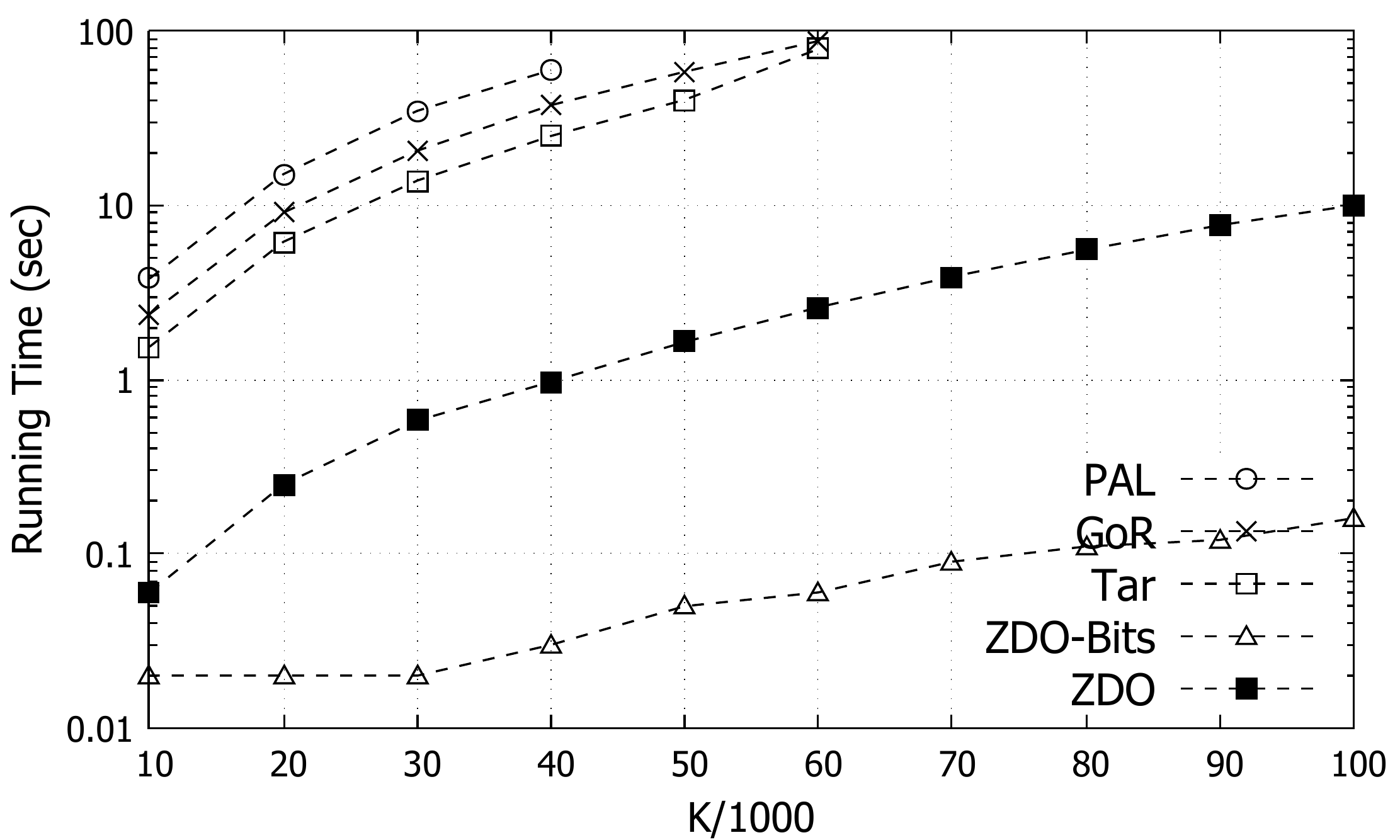}}
\caption{Running times for different algorithms on Bad-GoR family}
\label{f4}
\end{figure}
%
\clearpage
\subsection{Simple grids}
SPGRID generator \cite{cherkassky1996shortest} was used to produce these grid families (square grids, long grids, wide grids). Vertices of these graphs correspond to points in the plane with integer coordinates $[x,y]$, where $1 \leq x \leq X$ and $1 \leq y \leq Y$. These points are connected "forward" by arcs
$([x,y], [x+1,y])$, "up" by arcs $([x,y], [x,(y+1) $~\mbox{mod}~$ Y)])$, and "down" by arcs $([x,y],[x,(y-1) $~\mbox{mod}~$ Y])$, where $1 \leq x \leq X$ and $1 \leq y \leq Y$. 
Thus, "up" and "down" arcs with the same $x$ value form a doubly connected cycle called a layer.
There is a source vertex connected to all vertices with coordinates $[1,y]$, where $1 \leq y \leq Y$. 
Arc lengths are selected uniformly at random from the interval [0, 10000].
An artificial source is connected to the source vertex with zero-length arc, and to other vertices with a fixed arc length equals to $10^8$.

\paragraph{Square grids.}

This family represents the square grids (S-grids). Vertices of these grid networks correspond to points in the plane where $X = Y$. 
Table \ref{t3} presents results for different algorithms on this family. 
All algorithms except Pal perform few relaxation checks.
ZDO-Bits performs the fewest number of main relaxation checks.
Figure \ref{f5} presents the running time for different algorithms on this family.
ZDO and Tar are the fastest among others. The worst performance in this family is that of Pal.

\begin{table}[!htb]
\centering
\caption{\\ Number of relaxation checks per arc for different algorithms on S-grids ($X=Y$)\label{t3}}{
\begin{tabular}{|c||c|c|c|c|c|c|c|c|} \hline 
\multirow{2}{*}{$X=Y$}  & \multirow{2}{*}{Pal} & \multicolumn{2}{c|} {GoR} & \multirow{2}{*}{Tar} & \multicolumn{2}{c|}  {ZDO} & \multicolumn{2}{c|}  {ZDO-Bits} \\ \cline{3-4} \cline{6-7} \cline{8-9}
 &  & aux	& main &   &  aux	& main &  aux	& main \\ \hline \hline

64 & 21.852 &  4.028	& 1.936 & 2.023 & \cellcolor{lightgray}2.315	& \cellcolor{lightgray}1.938 & \cellcolor{lightgray}2.034 &\cellcolor{lightgray}1.183   \\ \hline

128 & 29.518 & 4.067	& 1.960 & 2.104 & \cellcolor{lightgray}2.406	& \cellcolor{lightgray}2.014 & \cellcolor{lightgray}2.120 &\cellcolor{lightgray}1.298    \\ \hline

256 & 53.707 & 4.084	& 1.969 & 2.136 & \cellcolor{lightgray}2.459	& \cellcolor{lightgray}2.058 & \cellcolor{lightgray}2.167 &\cellcolor{lightgray}1.253   \\ \hline

512 & 93.394 & 4.110	& 1.983 & 2.173 & \cellcolor{lightgray}2.503	& \cellcolor{lightgray}2.094 & \cellcolor{lightgray}2.203 &\cellcolor{lightgray}1.354  \\ \hline

1024& 331.492  & 4.111	& 1.982 & 2.211 & \cellcolor{lightgray}2.553	& \cellcolor{lightgray}2.133 & \cellcolor{lightgray}2.249 &\cellcolor{lightgray}1.296  \\ \hline

\end{tabular}}
\end{table}
\begin{figure}[!ht]
\centerline{\includegraphics[width=0.9\textwidth]{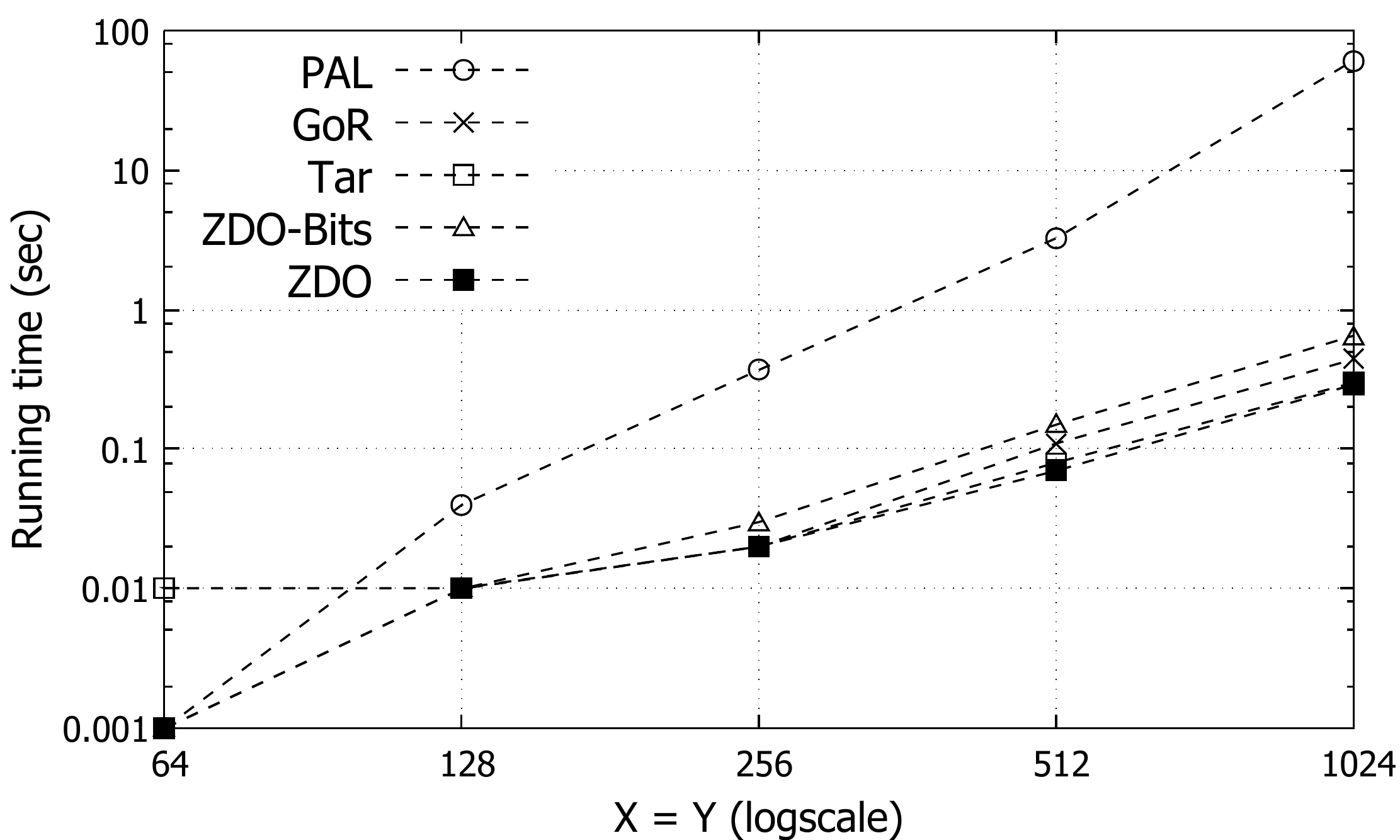}}
\caption{Running times for different algorithms on S-grids family}
\label{f5}
\end{figure}
%
\clearpage
\paragraph{Wide grids.}
The grids in this family (W-grids) have a fixed length $X=16$ and the width $Y$ grows with the problem size. Table \ref{t4} presents the number of relaxation checks and Figure \ref{f6} presents the running time for different algorithms on this family. ZDO-Bits achieves the fewest number of main relaxation checks. ZDO is the fastest among others. Although the total relaxation checks for ZDO is more than Tar,
the main checks for ZDO are less.

\begin{table}[!htb]
\centering
\caption{\\ Number of relaxation checks per arc for different algorithms on W-grids ($X=16$) \label{t4}}{
\begin{tabular}{|c||c|c|c|c|c|c|c|c|} \hline 
\multirow{2}{*}{Y}  & \multirow{2}{*}{Pal} & \multicolumn{2}{c|} {GoR} & \multirow{2}{*}{Tar} & \multicolumn{2}{c|}  {ZDO} & \multicolumn{2}{c|}  {ZDO-Bits} \\ \cline{3-4} \cline{6-7} \cline{8-9}
 &  & aux	& main &   &  aux	& main &  aux	& main \\ \hline \hline

512  & 3.049 & 3.868	& 1.850  & 1.939  & \cellcolor{lightgray}2.198	& \cellcolor{lightgray}1.797 &\cellcolor{lightgray}1.886 & \cellcolor{lightgray}1.165\\ \hline

1024 & 5.428 & 3.911	& 1.867  & 1.922 & \cellcolor{lightgray}2.168	& \cellcolor{lightgray}1.781 &\cellcolor{lightgray}1.877 & \cellcolor{lightgray}1.105\\ \hline

2048 & 3.987 & 3.909	& 1.867  & 1.931 & \cellcolor{lightgray}2.187	& \cellcolor{lightgray}1.792 &\cellcolor{lightgray}1.883 & \cellcolor{lightgray}1.152\\ \hline

4096 & 5.709 & 3.946	& 1.886  & 1.911 & \cellcolor{lightgray}2.139	& \cellcolor{lightgray}1.764 &\cellcolor{lightgray}1.876 & \cellcolor{lightgray}1.103\\ \hline

8192 & 2.411 & 3.868	& 1.849  & 1.944 & \cellcolor{lightgray}2.207	& \cellcolor{lightgray}1.803 &\cellcolor{lightgray}1.888 & \cellcolor{lightgray}1.165\\ \hline

16384 & 3.941 & 3.915	& 1.869  & 1.936 & \cellcolor{lightgray}2.193	& \cellcolor{lightgray}1.797 &\cellcolor{lightgray}1.888 & \cellcolor{lightgray}1.120\\ \hline

32768 & 2.427 & 3.901	& 1.865  & 1.946 & \cellcolor{lightgray}2.213	& \cellcolor{lightgray}1.808 &\cellcolor{lightgray}1.893 & \cellcolor{lightgray}1.170\\ \hline

\end{tabular}}
\end{table}

\begin{figure}[!ht]
\centerline{\includegraphics[width=0.9\textwidth]{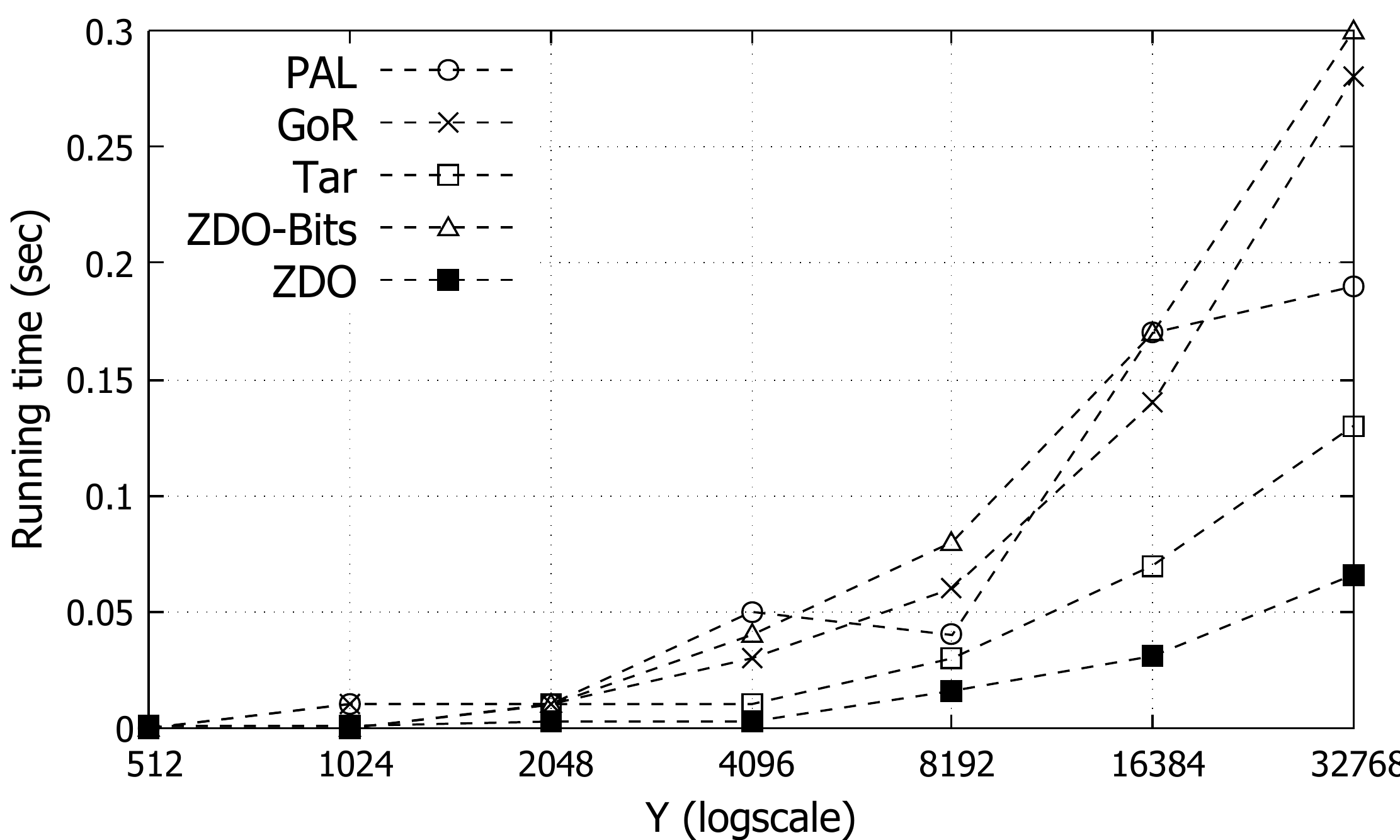}}
\caption{Running times for different algorithms on W-grids family}
\label{f6}
\end{figure}
%
\clearpage
\paragraph{Long grids.}
The grids in this family (L-grids) have a fixed width $Y=16$ and their length $X$ grows with the problem size. Table \ref{t5} presents the number of relaxation checks and Figure \ref{f7} presents the running time for different algorithms on this family. 
The best performance on this family is achieved by ZDO then Tar. The worst performance on this family concerning the number of relaxation checks and the running time is that of Pal, which solves six cases out of seven in less than 100 seconds.
\begin{table}[!htb]
\centering
\caption{\\ Number of relaxation checks per arc for different algorithms on L-grids ($Y=16$) \label{t5}}{
\begin{tabular}{|c||c|c|c|c|c|c|c|c|} \hline 
\multirow{2}{*}{X}  & \multirow{2}{*}{Pal} & \multicolumn{2}{c|} {GoR} & \multirow{2}{*}{Tar} & \multicolumn{2}{c|}  {ZDO} & \multicolumn{2}{c|}  {ZDO-Bits} \\ \cline{3-4} \cline{6-7} \cline{8-9}
 &  & aux	& main &   &  aux	& main &  aux	& main \\ \hline \hline

512 & 46.143 & 4.034	& 1.946 & 2.080  & \cellcolor{lightgray}2.387	& \cellcolor{lightgray}2.005 & \cellcolor{lightgray}2.108 &\cellcolor{lightgray}1.218 \\ \hline

1024 & 270.638 &  4.056	& 1.956  & 2.078 & \cellcolor{lightgray}2.381	& \cellcolor{lightgray}2.003 & \cellcolor{lightgray}2.107  &\cellcolor{lightgray}1.220 \\ \hline

2048 & 548.552 & 4.044	& 1.952  & 2.067 & \cellcolor{lightgray}2.370	& \cellcolor{lightgray}1.996 & \cellcolor{lightgray}2.093  &\cellcolor{lightgray}1.214 \\ \hline

4096 & 1087.042 & 4.062	& 1.959  & 2.082 & \cellcolor{lightgray}2.387	& \cellcolor{lightgray}2.009 & \cellcolor{lightgray}2.108  &\cellcolor{lightgray}1.221 \\ \hline

8192 & 280.250 & 4.050	& 1.955  & 2.077 & \cellcolor{lightgray}2.382	& \cellcolor{lightgray}2.005 & \cellcolor{lightgray}2.106  &\cellcolor{lightgray}1.220 \\ \hline

16384 & 2450.735 & 4.028& 1.944  & 2.081 & \cellcolor{lightgray}2.387	& \cellcolor{lightgray}2.009 & \cellcolor{lightgray}2.110  &\cellcolor{lightgray}1.222 \\ \hline
 
32768 &  &3.521	& 1.709 & 1.828 & \cellcolor{lightgray}2.062	& \cellcolor{lightgray}1.773 & \cellcolor{lightgray}1.791  &\cellcolor{lightgray}1.072\\ \hline

\end{tabular}}
\end{table}
%

\begin{figure}[!ht]
\centerline{\includegraphics[width=.9\textwidth]{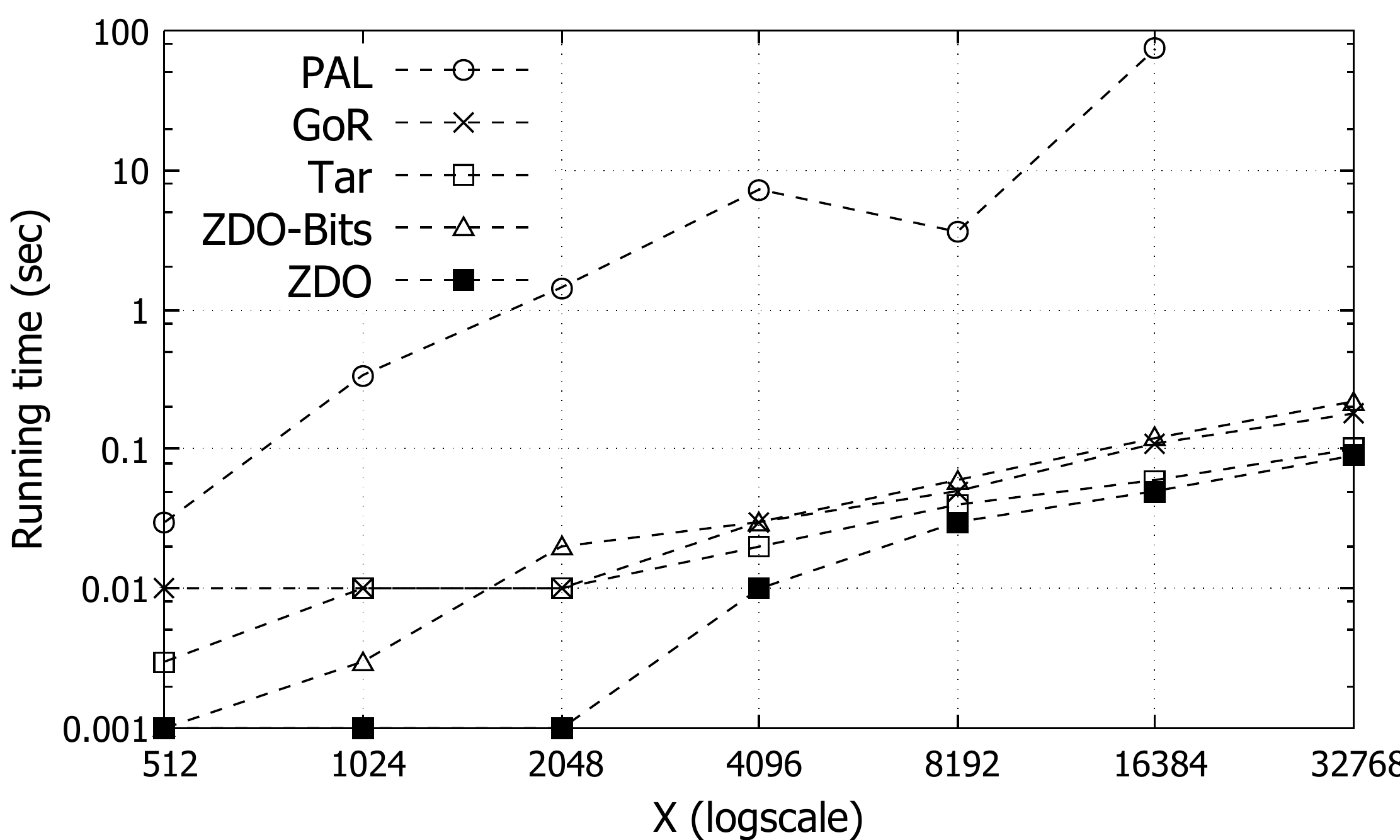}}
\caption{Running time for different algorithms on L-grids family}
\label{f7}
\end{figure}
%
\clearpage
\subsection{Hard grids}

These families are also produced by SPGRID generator.  As for the simple grids, the networks in this family 
consist of layers and the source connected to the vertices of the first layer. 
Each layer is a simple cycle plus a collection of arcs connecting randomly selected pairs
of vertices on the cycle. The length of the arcs inside a layer is small and non-negative. 
As for the simple grids, there are arcs from one layer to the next one. 
However, here, there are in addition arcs from lower to higher numbered layers. 
There is also artificial source connected to all vertices.

\paragraph{Positive hard grids.}
In this family (PH-grids), the inter-layer arcs have non-negative length 
selected uniformly at random from a wide range of integers. 
Additionally, the arc length from layer $x_1$ to layer $x_2$ is multiplied by $(x_2-x_1)^2$. 
Table \ref{t6} presents the number of relaxation checks and Figure \ref{f8} presents the running time for different algorithms on this family. ZDO achieves the best performance then Tar. The worst performance is that of Pal, which solves only four case out of six cases in less than 100 seconds.

\begin{table}[!htb]
\centering
\caption{\\ Number of relaxation checks per arc for different algorithms on PH-grids ($Y=32$)\label{t6}}{
\begin{tabular}{|c||c|c|c|c|c|c|c|c|} \hline 
\multirow{2}{*}{X}  & \multirow{2}{*}{Pal} & \multicolumn{2}{c|} {GoR} & \multirow{2}{*}{Tar} & \multicolumn{2}{c|}  {ZDO} & \multicolumn{2}{c|}  {ZDO-Bits} \\ \cline{3-4} \cline{6-7} \cline{8-9}
 &  & aux	& main &   &  aux	& main &  aux	& main \\ \hline \hline

256  & 341.936 & 14.092	 & 6.969 & 8.895  & \cellcolor{lightgray}8.373	& \cellcolor{lightgray}5.610 &\cellcolor{lightgray}11.309 &\cellcolor{lightgray}2.902 \\ \hline

512  & 650.544 & 14.331	 & 7.144 & 9.213 & \cellcolor{lightgray}8.232	& \cellcolor{lightgray}5.575 &\cellcolor{lightgray}11.213 &\cellcolor{lightgray}2.873 \\ \hline

1024  & 1156.602 & 15.633	 & 7.800 & 9.252 & \cellcolor{lightgray}8.242	& \cellcolor{lightgray}5.598 &\cellcolor{lightgray}11.294 &\cellcolor{lightgray}2.892 \\ \hline

2048  & 2358.277 & 14.443	 & 7.212 & 9.299 & \cellcolor{lightgray}8.313	& \cellcolor{lightgray}5.651 &\cellcolor{lightgray}11.405 &\cellcolor{lightgray}2.909 \\ \hline

4096  & & 14.467	 & 7.225 & 9.340 & \cellcolor{lightgray}8.304	& \cellcolor{lightgray}5.651 &\cellcolor{lightgray}11.410 &\cellcolor{lightgray}2.907 \\ \hline

8192  & & 14.194	 & 7.088 & 9.316 & \cellcolor{lightgray}8.297	& \cellcolor{lightgray}5.647 &\cellcolor{lightgray}11.396 &\cellcolor{lightgray}2.908 \\ \hline
 
\end{tabular}}
\end{table}
\begin{figure}[!ht]
\centerline{\includegraphics[width=0.9\textwidth]{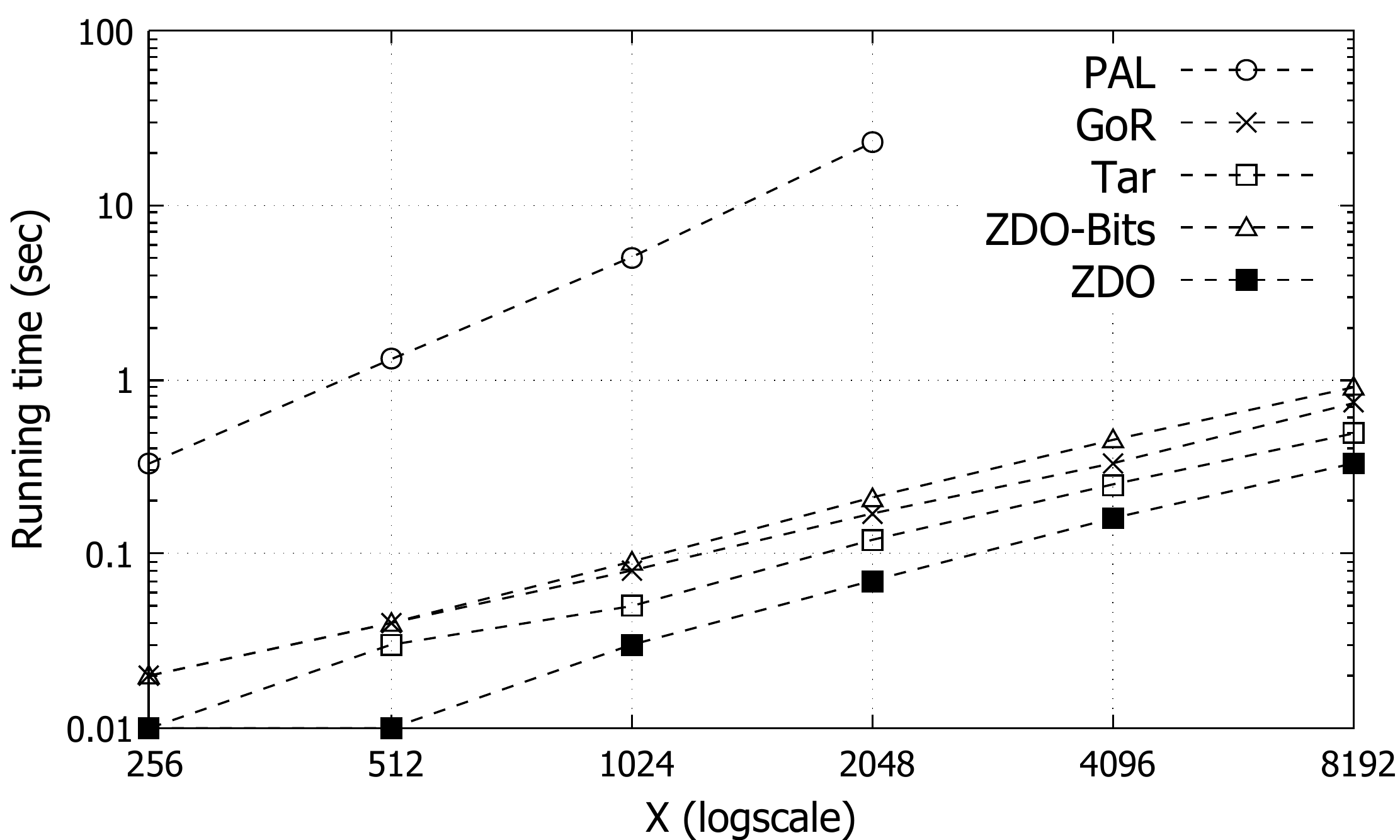}}
\caption{Running time for different algorithms on PH-grids family}
\label{f8}
\end{figure}
%

\clearpage
\paragraph{Negative hard grids.}
In this family (NH-grids), the inter-layer arcs have non-positive length
selected uniformly at random from a wide range of integers. 
Table \ref{t7} presents the number of relaxation checks and Figure \ref{f9} presents the running time for different algorithms on this family. GoR and ZDO are superior in this family. The worst performance is that of Pal, which solves only four out of six cases in less than 100 seconds.

\begin{table}[!htb]
\centering
\caption{\\ Number of relaxation checks per arc for different algorithms on NH-grids ($Y=32$)\label{t7}}{
\begin{tabular}{|c||c|c|c|c|c|c|c|c|} \hline 
\multirow{2}{*}{X}  & \multirow{2}{*}{Pal} & \multicolumn{2}{c|} {GoR} & \multirow{2}{*}{Tar} & \multicolumn{2}{c|}  {ZDO} & \multicolumn{2}{c|}  {ZDO-Bits} \\ \cline{3-4} \cline{6-7} \cline{8-9}
 &  & aux	& main &   &  aux	& main &  aux	& main \\ \hline \hline

256  & 336.765 & 12.778	 & 6.269 & 19.974  & \cellcolor{lightgray}17.614	& \cellcolor{lightgray}10.525 &\cellcolor{lightgray}24.546 &\cellcolor{lightgray}7.328 \\ \hline

512  & 642.307 & 13.669	 & 6.705 & 23.263 & \cellcolor{lightgray}18.694	& \cellcolor{lightgray}11.272 &\cellcolor{lightgray}26.488 &\cellcolor{lightgray}7.948 \\ \hline

1024  & 1257.516 & 14.212	 & 6.978 & 28.991 & \cellcolor{lightgray}18.781	& \cellcolor{lightgray}11.387 &\cellcolor{lightgray}26.756 &\cellcolor{lightgray}8.063 \\ \hline

2048  & 2582.655 & 15.434	 & 7.589 & 33.855 & \cellcolor{lightgray}19.157	& \cellcolor{lightgray}11.632 &\cellcolor{lightgray}27.394 &\cellcolor{lightgray}8.263 \\ \hline

4096  & & 16.043	 & 7.893 & 44.509 & \cellcolor{lightgray}19.263	& \cellcolor{lightgray}11.709 &\cellcolor{lightgray}27.582 &\cellcolor{lightgray}8.325 \\ \hline

8192  & & 16.544	 & 8.142 & 49.035 & \cellcolor{lightgray}19.271	& \cellcolor{lightgray}11.723 &\cellcolor{lightgray}27.619 &\cellcolor{lightgray}8.335 \\ \hline

\end{tabular}}
\end{table}
\begin{figure}[!ht]
\centerline{\includegraphics[width=0.9\textwidth]{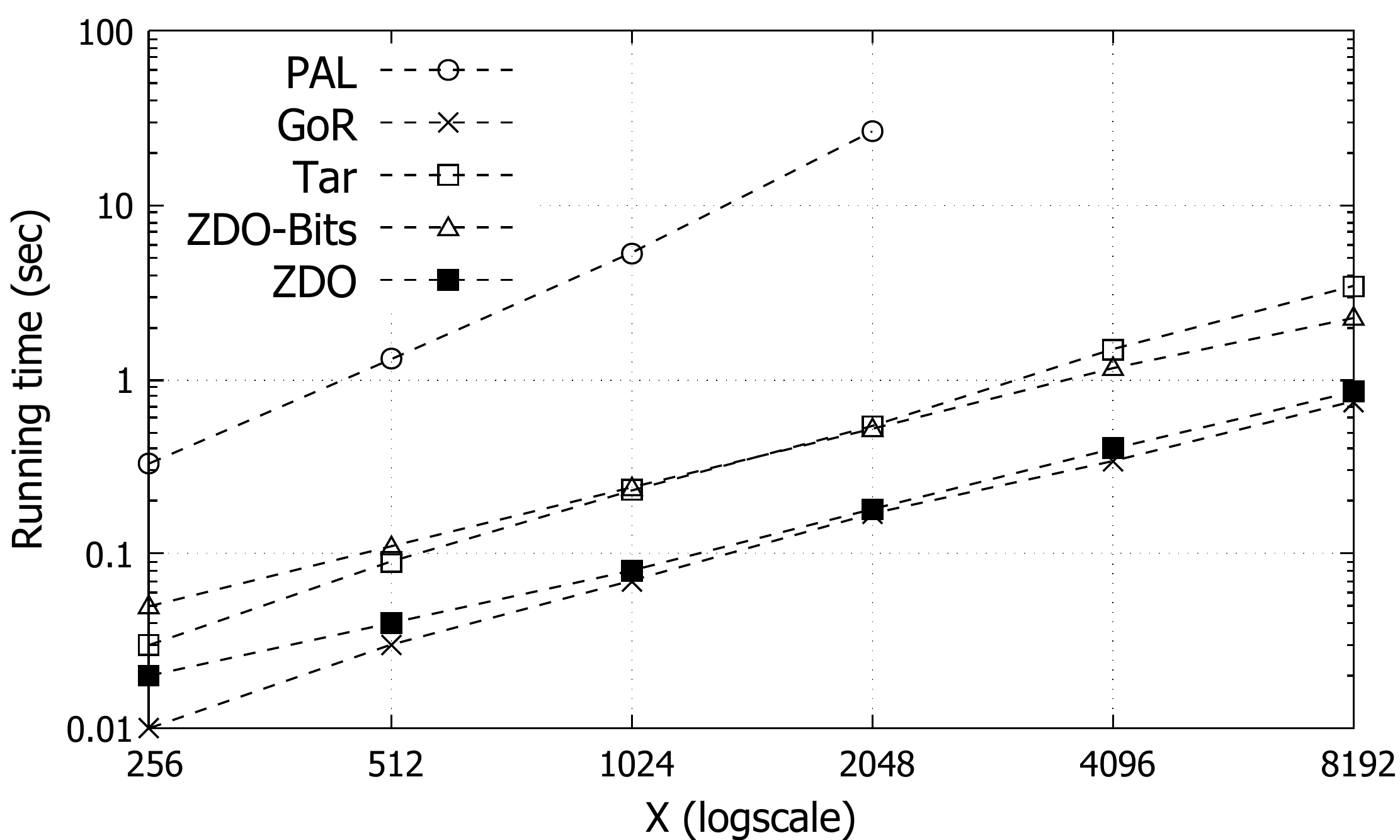}}
\caption{Running time for different algorithms on NH-grids family}
\label{f9}
\end{figure}
%
\clearpage
\subsection{Random graphs}
These families are produced by the SPRAND generator. The graphs are constructed as a Hamiltonian cycle in addition to arcs with distinct random end points. The lengths of these arcs are selected uniformly at random from the interval $[0, 10000]$.

\paragraph{Sparse random graphs.}
The graphs in this family (S-rand) are sparse random  graphs with $m = 4n$. Table \ref{t8} presents the number of relaxation checks and Figure \ref{f10} presents the running time for different algorithms on this family. ZDO-Bits performs the fewest number of main relaxation checks. ZDO is the fastest algorithm in this family.

\begin{table}[!htb]
\centering
\caption{\\ Number of relaxation checks per arc for different algorithms on S-rand family\label{t8}}{
\begin{tabular}{|c||c|c|c|c|c|c|c|c|} \hline 
\multirow{2}{*}{n}  & \multirow{2}{*}{Pal} & \multicolumn{2}{c|} {GoR} & \multirow{2}{*}{Tar} & \multicolumn{2}{c|}  {ZDO} & \multicolumn{2}{c|}  {ZDO-Bits} \\ \cline{3-4} \cline{6-7} \cline{8-9}
 &  & aux	& main &   &  aux	& main &  aux	& main \\ \hline \hline

8192 & 10.560 & 11.358	 & 5.450 & 5.766  & \cellcolor{lightgray}5.611	& \cellcolor{lightgray}5.090 &\cellcolor{lightgray}6.659 &\cellcolor{lightgray}1.971 \\ \hline

16384 & 11.966 & 12.160	& 5.835	& 6.297	& \cellcolor{lightgray}6.143	&\cellcolor{lightgray}5.575	&\cellcolor{lightgray}7.259	&\cellcolor{lightgray}2.136 \\ \hline

32768 & 12.406 & 12.454	 & 5.979 & 6.474 & \cellcolor{lightgray}6.336	& \cellcolor{lightgray}5.738 &\cellcolor{lightgray}7.511 &\cellcolor{lightgray}2.211 \\ \hline

65536 & 14.520 & 13.685	 & 6.583 & 7.094 & \cellcolor{lightgray}6.941	&  \cellcolor{lightgray}6.308 &\cellcolor{lightgray}8.193 &\cellcolor{lightgray}2.394 \\ \hline
 
131072 & 16.425 & 14.815 & 7.131 & 7.709 & \cellcolor{lightgray}7.509	& \cellcolor{lightgray}6.821 &\cellcolor{lightgray}8.889 &\cellcolor{lightgray}2.595 \\ \hline

262144 & 17.057 & 15.326 & 7.377 & 7.871 & \cellcolor{lightgray}7.693	& \cellcolor{lightgray}6.971 &\cellcolor{lightgray}9.135 &\cellcolor{lightgray}2.671 \\ \hline

524288 & 18.065 & 15.882 & 7.646 & 8.265 &  \cellcolor{lightgray}8.082	& \cellcolor{lightgray}7.332 &\cellcolor{lightgray}9.577 &\cellcolor{lightgray}2.794 \\ \hline
 
1048576 & 19.335 & 16.564 & 7.980 & 8.549 &  \cellcolor{lightgray}8.353	& \cellcolor{lightgray}7.584 &\cellcolor{lightgray}9.884 &\cellcolor{lightgray}2.876 \\ \hline
\end{tabular}}
\end{table}
\begin{figure}[!ht]
\centerline{\includegraphics[width=0.9\textwidth]{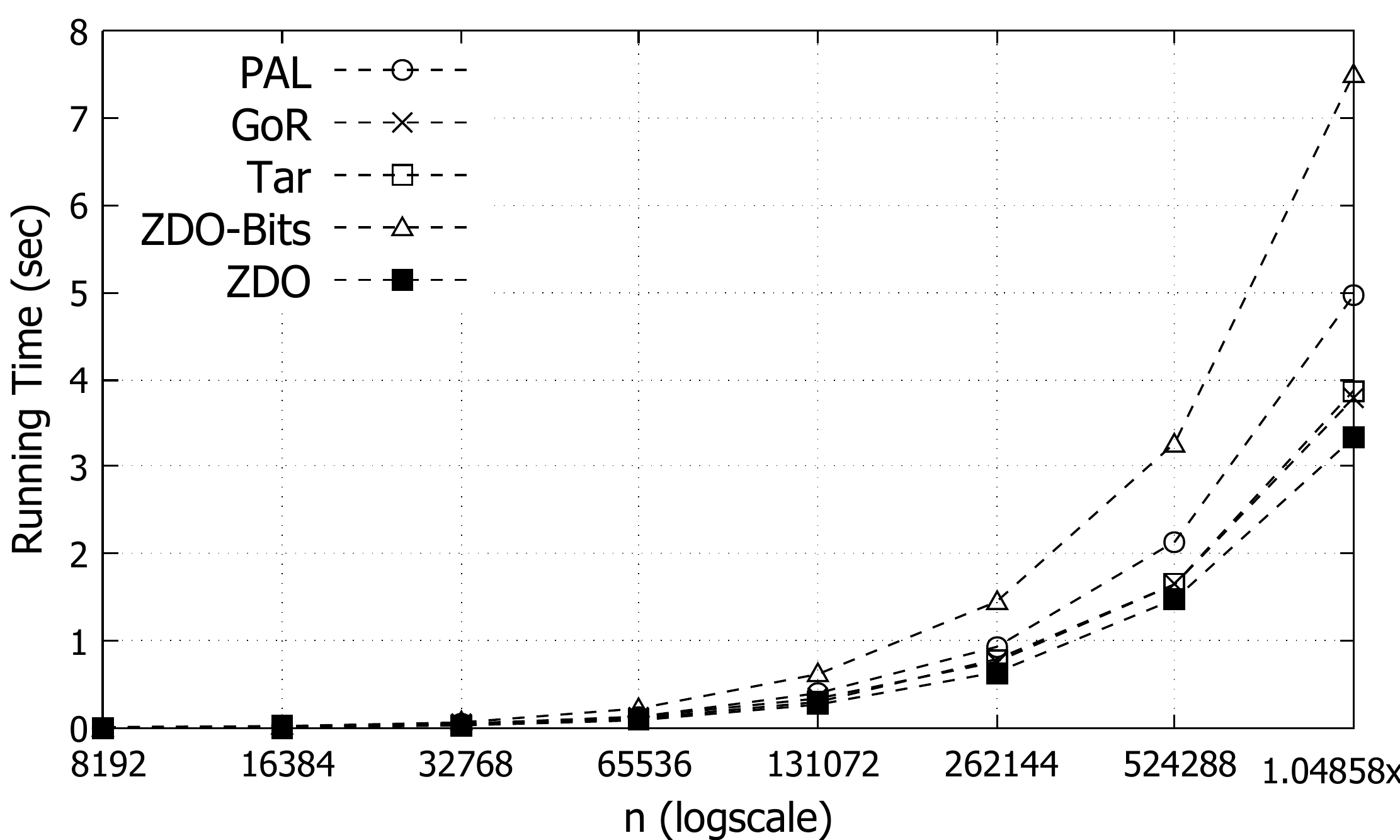}}
\caption{Running time for different algorithms on S-rand family}
\label{f10}
\end{figure}
%
\clearpage
\paragraph{Dense random graphs.}
The graphs in this family (D-rand) are dense random graphs with $m=n^2/4$. 
Table \ref{t9} presents the number of relaxation checks for different algorithms on this family
and Figure \ref{f11} presents their running times. 
ZDO-Bits achieves the fewest number of main relaxation checks. 
ZDO is also the fastest algorithm in this family, Tar comes next, GoR is slower than Pal, and ZDO-Bits is the slowest.

\begin{table}[!htb]
\centering
\caption{\\ Number of relaxation checks per arc for different algorithms on D-rand family\label{t9}}{
\begin{tabular}{|c||c|c|c|c|c|c|c|c|} \hline 
\multirow{2}{*}{n}  & \multirow{2}{*}{Pal} & \multicolumn{2}{c|} {GoR} & \multirow{2}{*}{Tar} & \multicolumn{2}{c|}  {ZDO} & \multicolumn{2}{c|}  {ZDO-Bits} \\ \cline{3-4} \cline{6-7} \cline{8-9}
 &  & aux	& main &   &  aux	& main &  aux	& main \\ \hline \hline

512  & 5.133 & 5.542	 & 2.632 & 2.992  &  \cellcolor{lightgray}2.351	&  \cellcolor{lightgray}1.973 &\cellcolor{lightgray}5.333 &\cellcolor{lightgray}0.306 \\ \hline

1024 & 5.250 & 5.681	 & 2.702 & 2.924 &  \cellcolor{lightgray}2.308	&  \cellcolor{lightgray}1.911 &\cellcolor{lightgray}5.660 &\cellcolor{lightgray}0.282 \\ \hline

2048  & 4.970 & 5.168	 & 2.467 & 2.894 &  \cellcolor{lightgray}2.231&  \cellcolor{lightgray}1.863 &\cellcolor{lightgray}5.806 &\cellcolor{lightgray}0.268 \\ \hline

4096 & 5.196 & 5.666	 & 2.765 & 3.344 &  \cellcolor{lightgray}2.641 &  \cellcolor{lightgray}2.278 &\cellcolor{lightgray}6.556 &\cellcolor{lightgray}0.263 \\ \hline
\end{tabular}}
\end{table}
\begin{figure}[!ht]
\centerline{\includegraphics[width=0.9\textwidth]{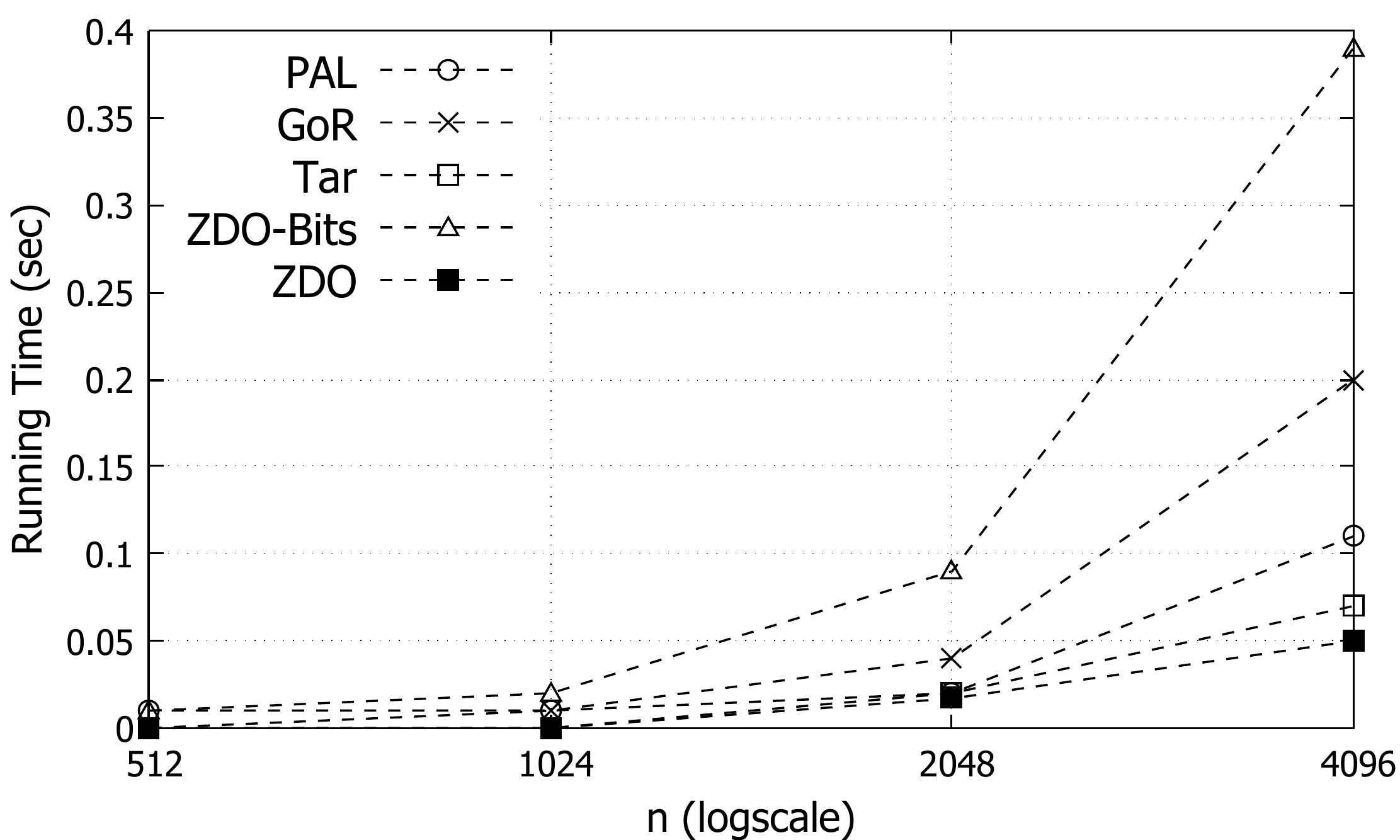}}
\caption{Running time for different algorithms on D-rand family}
\label{f11}
\end{figure}
%
\clearpage
\paragraph{Random graphs with potentials.}
Graphs in this family are also produced by SPRAND generator. They are random graphs with negative arcs but without negative cycles. To do this, there are two variables that control an arc length ($l'$): the length function ($l$) that is produced uniformly from the interval $[0,10000]$ for each arc, and the potential value $p$ that is produced uniformly from the interval $[0,P]$ for each vertex. 
The length of an arc $(u,v)$ is determined as $l'(u,v) = l(u,v) + p(u) - p(v)$.
The graphs in this family (P-rand) have fixed $n= 131072$ and $m=524288$. The value of $P$ is picked uniformly from  $0$ to $10^6$; changing the potential limit $P$ controls the number of negative arcs in the network. For the same seed, when we use $P=0$ the percentage of negative arcs is $0\%$, when $P=10^3$ the percentage of negative arcs is about $25\%$, and when $P=10^6$ the percentage of negative arcs is about $50\%$. Table \ref{pt1} presents the number of relaxation checks for different algorithms on this family.
For GoR, the number of checks slightly varies as $P$ changes. For the remaining algorithms, the number of operations does not change when changing $P$. Figure \ref{p1} presents the running time for different algorithms on this family. 
ZDO is also the fastest algorithm, while ZDO-Bits is the slowest in this family.
%

\begin{table}[!htb]
\centering
\caption{\\ Number of relaxation checks per arc for different algorithms on P-rand family (n=131072 , m=524288)\label{pt1}}{
\begin{tabular}{|c||c|c|c|c|c|c|c|c|} \hline 
\multirow{2}{*}{$P$}  & \multirow{2}{*}{Pal} & \multicolumn{2}{c|} {GoR} & \multirow{2}{*}{Tar} & \multicolumn{2}{c|}  {ZDO} & \multicolumn{2}{c|}  {ZDO-Bits} \\ \cline{3-4} \cline{6-7} \cline{8-9}
 &  & aux	& main &   &  aux	& main &  aux	& main \\ \hline \hline

0  & 16.425 & 14.815	 & 7.131 & 7.709  & \cellcolor{lightgray}7.509	& \cellcolor{lightgray}6.821 &\cellcolor{lightgray}8.889	&\cellcolor{lightgray}2.595 \\ \hline

1000  & 16.425 & 14.714	 & 7.089 & 7.709 & \cellcolor{lightgray}7.509	& \cellcolor{lightgray}6.821 &\cellcolor{lightgray}8.889	&\cellcolor{lightgray}2.595 \\ \hline

5000  & 16.425 & 15.091	 & 7.271 & 7.709 & \cellcolor{lightgray}7.509	& \cellcolor{lightgray}6.821 &\cellcolor{lightgray}8.889	&\cellcolor{lightgray}2.595 \\ \hline

10000  & 16.425 & 15.033	 & 7.245 & 7.709 & \cellcolor{lightgray}7.509	& \cellcolor{lightgray}6.821 &\cellcolor{lightgray}8.889	&\cellcolor{lightgray}2.595 \\ \hline

100000  & 16.425 & 15.297	 & 7.375 & 7.709 & \cellcolor{lightgray}7.509	& \cellcolor{lightgray}6.821 &\cellcolor{lightgray}8.889	&\cellcolor{lightgray}2.595 \\ \hline

1000000  & 16.425 & 15.006	 & 7.232 & 7.709 & \cellcolor{lightgray}7.509	& \cellcolor{lightgray}6.821 &\cellcolor{lightgray}8.889	&\cellcolor{lightgray}2.595 \\ \hline

\end{tabular}}
\end{table}

\begin{figure}[!ht]
\centerline{\includegraphics[width=0.9\textwidth]{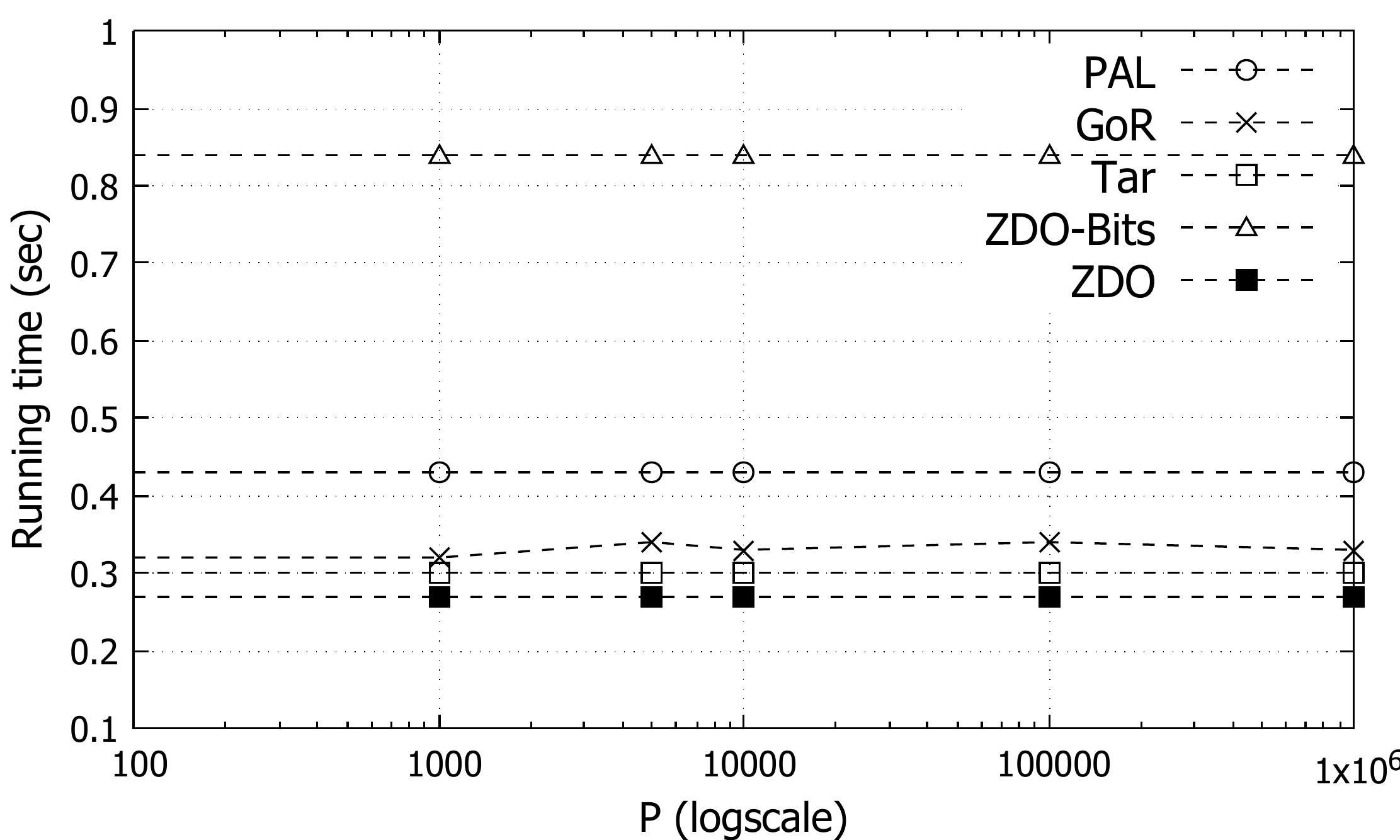}}
\caption{Effect of changing P on P-rand family (n=131072, m=524288)}
\label{p1}
\end{figure} 

\clearpage
\paragraph{Dense-to-sparse random graphs with potentials.}
This part studies the effect of changing the graphs from dense to sparse while maintaining the percentage of the negative arcs using potential values.
The graphs in this family (PD2S-rand) have a fixed number of arcs $m=10^7$ but the number of vertices $n$ changes from $10^2$ to $10^5$. 
The length function $l$ is produced uniformly from the interval $[0,10000]$ for each arc, and the potential value $p$ is produced uniformly from the interval $[0,10^6]$ for each vertex. This provides graphs with about $50\%$ of their arcs negative.
Table \ref{pt2} presents the number of relaxation checks per arc for different algorithms on this family and Figure \ref{p2} presents their running times. 
ZDO-Bits performs the fewest main relaxation checks. 
ZDO is the fastest then Tar, GoR is slower than Pal, while ZDO-Bits is the slowest.%

\begin{table}[!ht]
\centering
\caption{\\ Number of relaxation checks per arc for different algorithms on PD2S-rand family (m=$10^7$)\label{pt2}}{
\begin{tabular}{|c||c|c|c|c|c|c|c|c|} \hline 
\multirow{2}{*}{n}  & \multirow{2}{*}{Pal} & \multicolumn{2}{c|} {GoR} & \multirow{2}{*}{Tar} & \multicolumn{2}{c|}  {ZDO} & \multicolumn{2}{c|}  {ZDO-Bits} \\ \cline{3-4} \cline{6-7} \cline{8-9}
 &  & aux	& main &   &  aux	& main &  aux	& main \\ \hline \hline

100 & 1.006 & 2.002 & 1.000 & 1.008  & \cellcolor{lightgray}1.000	& \cellcolor{lightgray}1.000 &\cellcolor{lightgray}3.791 &\cellcolor{lightgray}0.495 \\ \hline

1000  & 1.576 & 4.025 & 2.000 & 1.344 & \cellcolor{lightgray}1.075	& \cellcolor{lightgray}1.026 &\cellcolor{lightgray}4.300 &\cellcolor{lightgray}0.304 \\ \hline

10000  & 6.900 & 8.312	 & 4.003 & 4.043 & \cellcolor{lightgray}3.259	& \cellcolor{lightgray}2.810 &\cellcolor{lightgray}7.678 &\cellcolor{lightgray}0.260 \\ \hline

100000  & 6.911 & 7.602	 & 3.567 & 3.575 & \cellcolor{lightgray}2.659	& \cellcolor{lightgray}2.174 &\cellcolor{lightgray}6.044 &\cellcolor{lightgray}0.331  \\ \hline
\end{tabular}}
\end{table}
\begin{figure}[!ht]
\centerline{\includegraphics[width=0.9\textwidth]{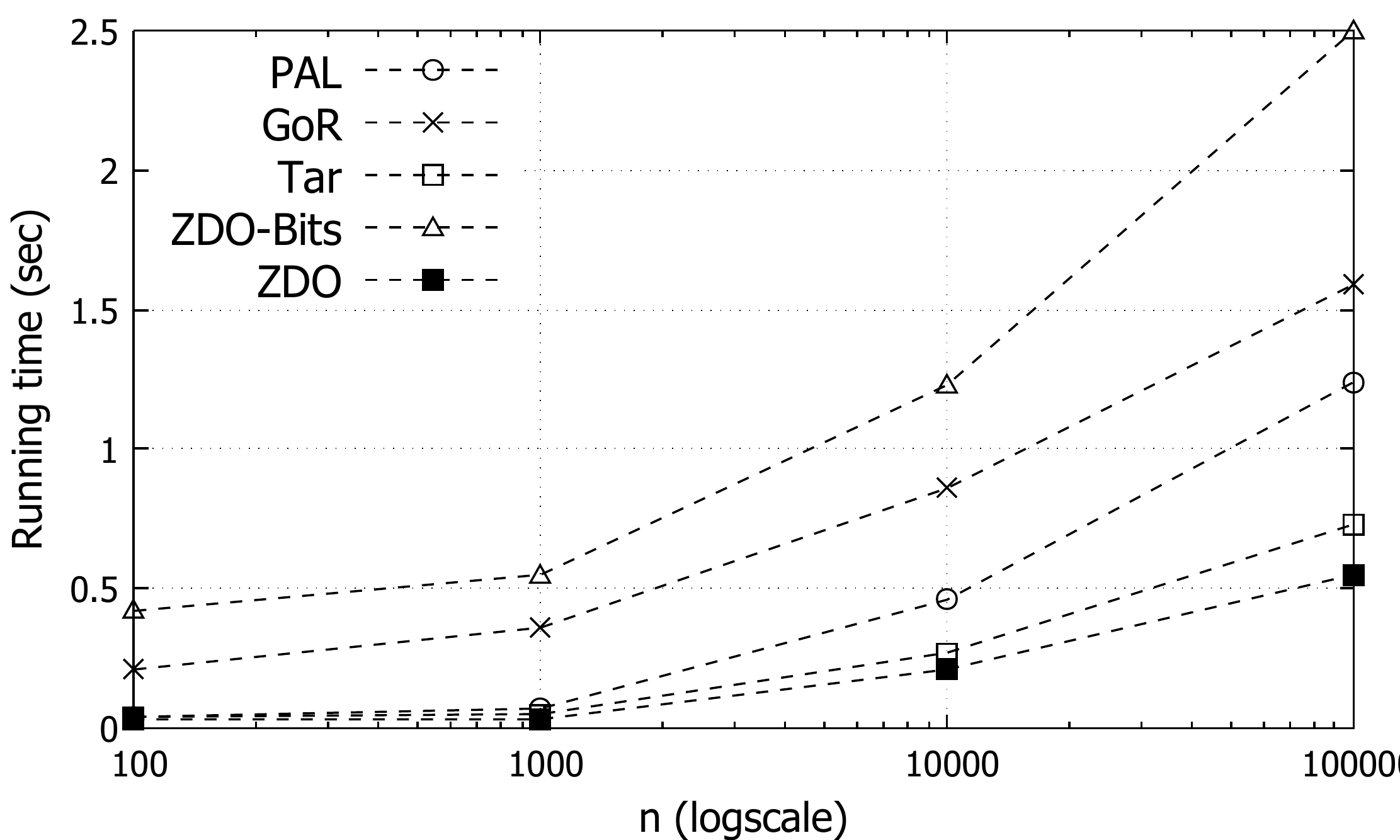}}
\caption{Running time for different algorithms on PD2S-rand family (m=$10^7$) }
\label{p2}
\end{figure}
%
\clearpage
\paragraph{Random graphs with artificial source and potentials.}

The graphs in this family (PS-rand) are the same as the graphs in PD2S-rand family but with an additional artificial source. The artificial source is connected to the original source by a zero-length arc and to the other vertices by very long arcs ($10^8$). 
Table \ref{pt3} presents the number of relaxation checks per arc and Figure \ref{p3} presents the running time for different algorithms on this family. 
Consistently, ZDO-Bits is the best in terms of the number of main relaxation checks and ZDO is the fastest.

\begin{table}[!htb]
\centering
\caption{\\ Number of relaxation checks per arc for different algorithms on PS-rand family \label{pt3}}{
\begin{tabular}{|c||c|c|c|c|c|c|c|c|} \hline 
\multirow{2}{*}{n}  & \multirow{2}{*}{Pal} & \multicolumn{2}{c|} {GoR} & \multirow{2}{*}{Tar} & \multicolumn{2}{c|}  {ZDO} & \multicolumn{2}{c|}  {ZDO-Bits} \\ \cline{3-4} \cline{6-7} \cline{8-9}
 &  & aux	& main &   &  aux	& main &  aux	& main \\ \hline \hline

100  & 2.000 & 2.480	 & 1.000 & 2.000  & \cellcolor{lightgray}1.000	& \cellcolor{lightgray}1.000 &\cellcolor{lightgray}3.203 &\cellcolor{lightgray}0.497 \\ \hline

1000 & 2.324 & 2.575	 & 1.200 & 2.324 & \cellcolor{lightgray}1.069	& \cellcolor{lightgray}1.019 &\cellcolor{lightgray}3.527 &\cellcolor{lightgray}0.388 \\ \hline

10000  & 6.919 & 7.922	 & 3.832 & 4.905 & \cellcolor{lightgray}3.927	& \cellcolor{lightgray}3.460 &\cellcolor{lightgray}8.365 &\cellcolor{lightgray}0.573 \\ \hline

100000  & 7.640 & 8.523	 & 4.018 & 4.533 & \cellcolor{lightgray}3.330	& \cellcolor{lightgray}2.744 &\cellcolor{lightgray}7.400 &\cellcolor{lightgray}0.720 \\ \hline
\end{tabular}}
\end{table}
\begin{figure}[!ht]
\centerline{\includegraphics[width=0.9\textwidth]{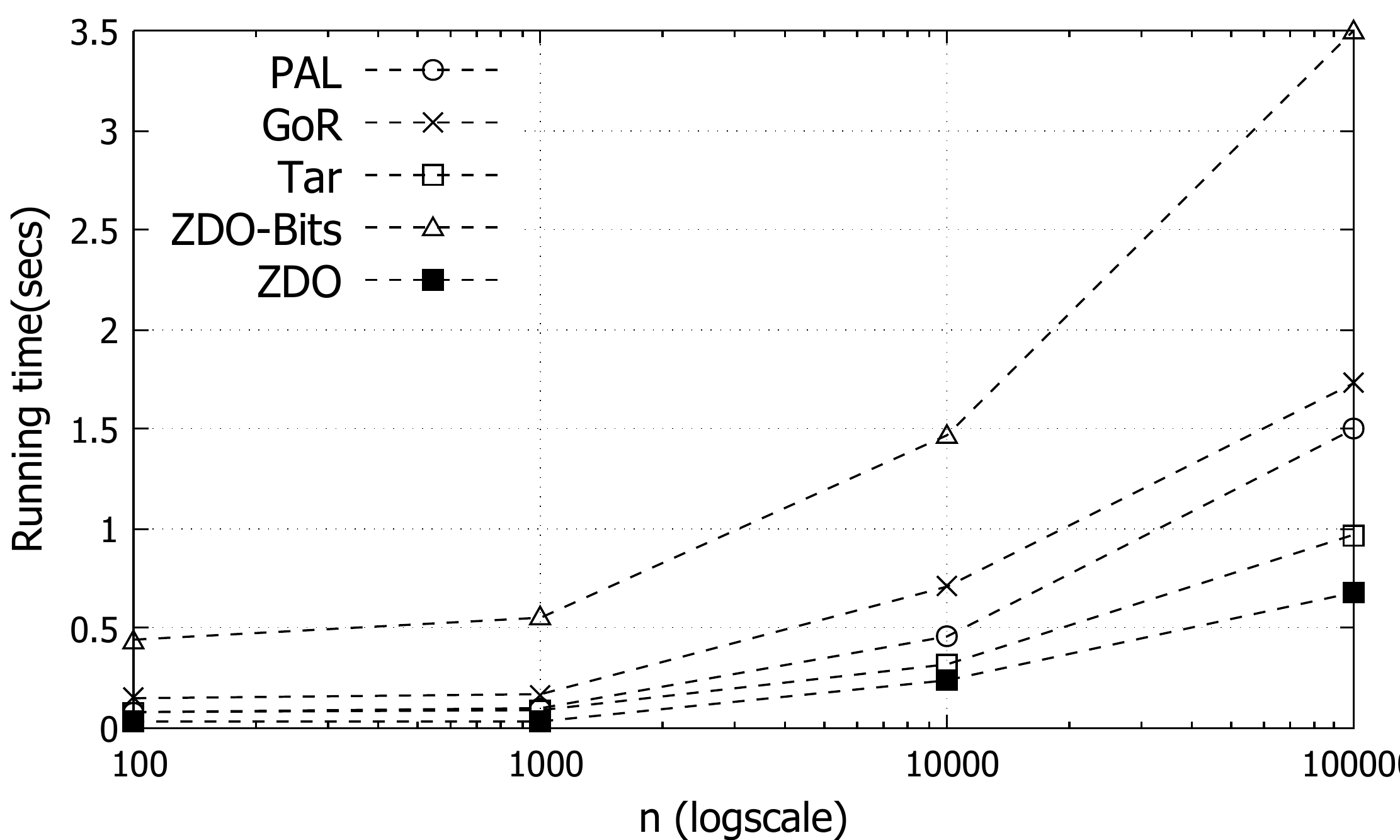}}
\caption{Running time for different algorithms on PS-rand family}
\label{p3}
\end{figure}
%
\clearpage
\paragraph{Random graphs with Hamiltonian cycle and potentials.}

The graphs in this family (PC-rand) are the same as the graphs in PD2S-rand family but with a Hamiltonian cycle that connects all vertices together. In our experiments, we set arc lengths on the cycle to $1$ and pick the others uniformly at random from $[0,10000]$ as before.
Table \ref{pt4} presents the number of relaxation checks per arc and 
Figure \ref{p4} presents the running time for different algorithms on this family. 
ZDO is again the fastest algorithm in this family. 

\begin{table}[!htb]
\centering
\caption{\\ Number of relaxation checks per arc for different algorithms on PC-rand family (m=$10^7$)\label{pt4}}{
\begin{tabular}{|c||c|c|c|c|c|c|c|c|} \hline 
\multirow{2}{*}{n}  & \multirow{2}{*}{Pal} & \multicolumn{2}{c|} {GoR} & \multirow{2}{*}{Tar} & \multicolumn{2}{c|}  {ZDO} & \multicolumn{2}{c|}  {ZDO-Bits} \\ \cline{3-4} \cline{6-7} \cline{8-9}
 &  & aux	& main &   &  aux	& main &  aux	& main \\ \hline \hline

100  & 1.008 & 2.004	 & 1.000 & 1.008  & \cellcolor{lightgray}1.000	& \cellcolor{lightgray}1.000 &\cellcolor{lightgray}3.820 &\cellcolor{lightgray}0.491 \\ \hline

1000  & 1.566 & 2.044	 & 1.000 & 1.402 & \cellcolor{lightgray}1.114	& \cellcolor{lightgray}1.044 &\cellcolor{lightgray}4.403 &\cellcolor{lightgray}0.291 \\ \hline

10000  & 5.868 & 4.921	 & 2.390 & 3.718 & \cellcolor{lightgray}2.959	& \cellcolor{lightgray}2.576 &\cellcolor{lightgray}7.030 &\cellcolor{lightgray}0.262 \\ \hline

100000  & 11.428 & 9.455	 & 4.570 & 4.795 & \cellcolor{lightgray}3.788	& \cellcolor{lightgray}3.184 &\cellcolor{lightgray}7.900 &\cellcolor{lightgray}0.350 \\ \hline
\end{tabular}}
\end{table}
\begin{figure}[!ht]
\centerline{\includegraphics[width=0.9\textwidth]{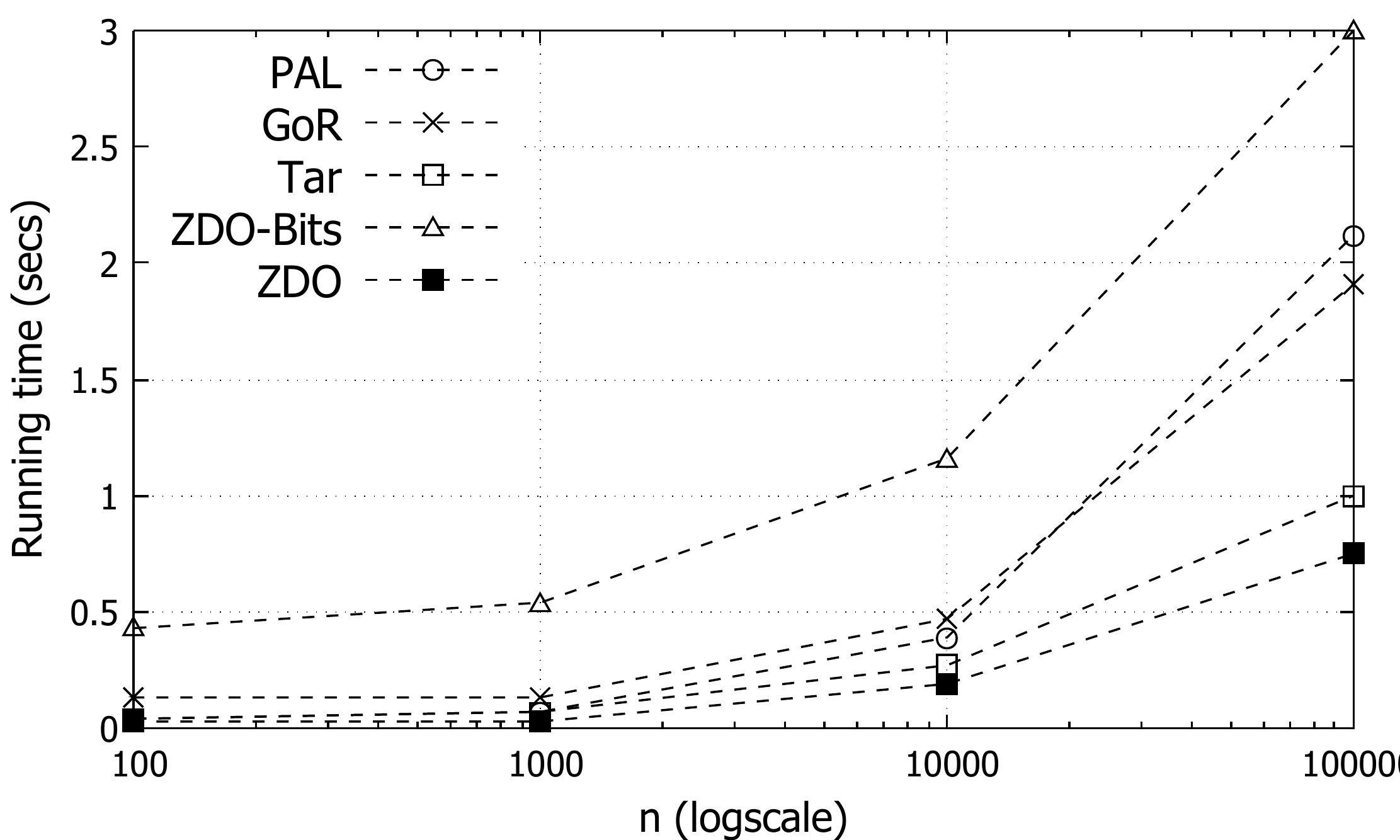}}
\caption{Running time for different algorithms on PC-rand family (m=$10^7$)}
\label{p4}
\end{figure}
%
\clearpage
\subsection{Acyclic graphs}
All graphs in this family are produced by the SPACYC generator. The vertices are numbered from 1 to n, and there is a path of arcs $(i, i + 1)$ for $1 \leq i < n$. These arcs are called the path arcs. Additional arcs are generated by picking two distinct vertcies at random and creating an arc from the lower to the higher numbered vertex. The lengths of the additional arcs are selected uniformly at random from the interval $[L,U]$. 

\paragraph{Fully positive arc lengths.}

For the positive acyclic family (FP-acyc) , the length of the path arcs is set to $1$ and the other arc lengths are selected uniformly at random from the interval $[0,10000]$. 
Table \ref{t10} presents the number of relaxation checks per arc and
Figure \ref{f12} presents the running time for different algorithms.
ZDO and Tar are the champions in this family.
 
\begin{table}[!htb]
\centering
\caption{\\ Number of relaxation checks per arc for different algorithms on FP-acyc family\label{t10}}{
\begin{tabular}{|c||c|c|c|c|c|c|c|c|} \hline 
\multirow{2}{*}{n}  & \multirow{2}{*}{Pal} & \multicolumn{2}{c|} {GoR} & \multirow{2}{*}{Tar} & \multicolumn{2}{c|}  {ZDO} & \multicolumn{2}{c|}  {ZDO-Bits} \\ \cline{3-4} \cline{6-7} \cline{8-9}
 &  & aux	& main &   &  aux	& main &  aux	& main \\ \hline \hline

8192  & 10.236 & 12.684	 & 6.106 & 6.330  & \cellcolor{lightgray}6.651	& \cellcolor{lightgray}5.549 &\cellcolor{lightgray}7.671 &\cellcolor{lightgray}0.694 \\ \hline

16384  & 11.940 & 14.703	 & 7.103 & 6.833 & \cellcolor{lightgray}7.246	& \cellcolor{lightgray}6.041 &\cellcolor{lightgray}8.167 &\cellcolor{lightgray}0.710 \\ \hline

32768  & 12.297 & 14.834	 & 7.168 & 7.126 & \cellcolor{lightgray}7.556	& \cellcolor{lightgray}6.283 &\cellcolor{lightgray}8.526 &\cellcolor{lightgray}0.747 \\ \hline

65536  & 13.813 & 17.719	 & 8.592 & 8.042 & \cellcolor{lightgray}8.473	& \cellcolor{lightgray}7.212 &\cellcolor{lightgray}9.477 &\cellcolor{lightgray}0.791 \\ \hline
 
131072  & 14.123 & 17.032	 & 8.241 & 8.085 & \cellcolor{lightgray}8.600	& \cellcolor{lightgray}7.247 &\cellcolor{lightgray}9.512 &\cellcolor{lightgray}0.791 \\ \hline

\end{tabular}}
\end{table}
\begin{figure}[!ht]
\centerline{\includegraphics[width=0.9\textwidth]{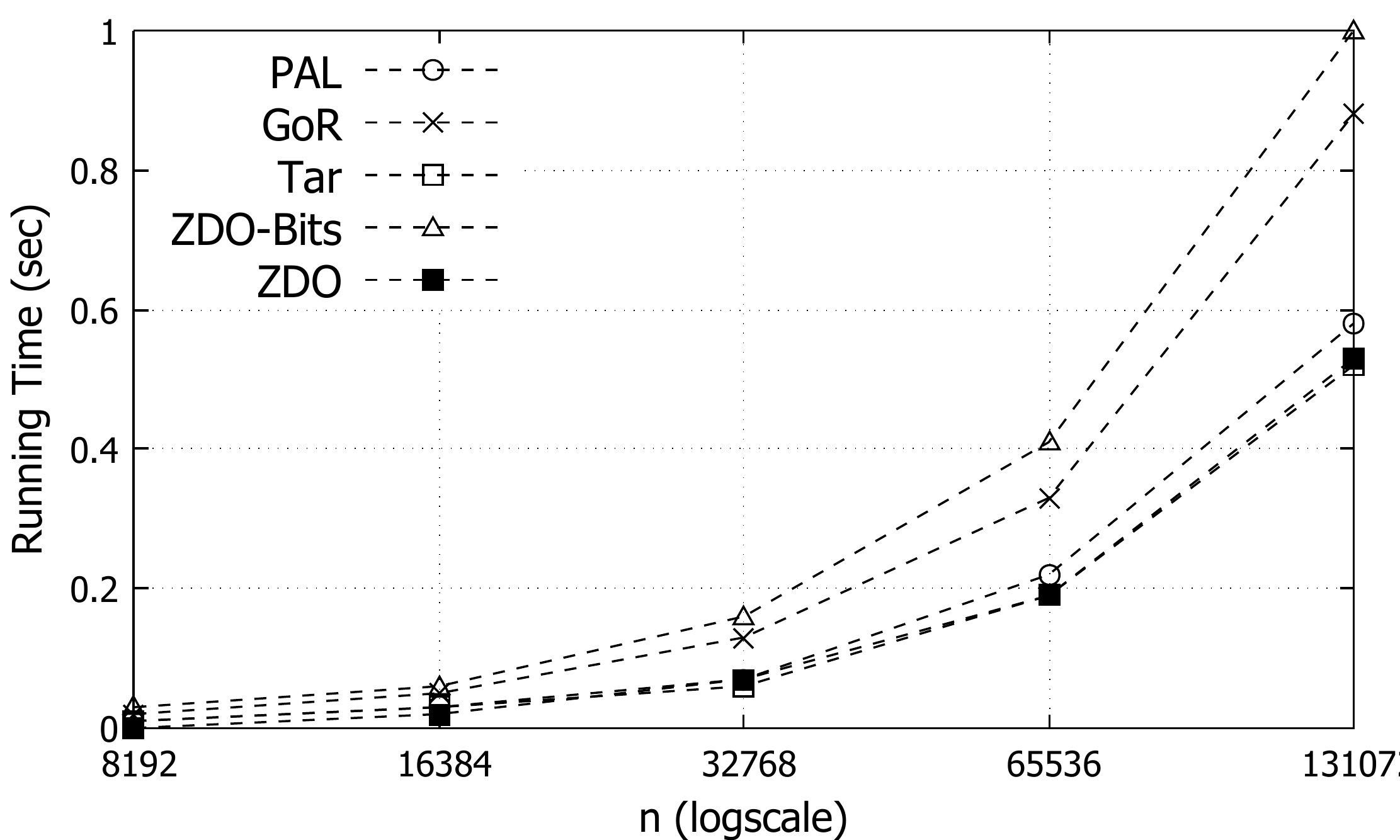}}
\caption{Running time for different algorithms on FP-acyc family}
\label{f12}
\end{figure}

\clearpage
\paragraph{Fully negative arc lengths.}
For the negative acyclic (FN-acyc) family, the length of the path arcs is set to $-1$ and the other arc lengths are selected uniformly at random from the interval $[-10000,0]$. 
Table \ref{t11} presents the number of relaxation checks per arc and 
Figure \ref{f13} presents the running time for different algorithms.
GoR is notably superior in this family then comes ZDO and Tar.
GoR solves the problem in one DFS pass and performs very few relaxation checks.
ZDO-Bits beats Pal in this family.   

\begin{table}[!htb]
\centering
\caption{\\ Number of relaxation checks per arc for different algorithms on FN-acyc family\label{t11}}{
\begin{tabular}{|c||c|c|c|c|c|c|c|c|} \hline 
\multirow{2}{*}{n}  & \multirow{2}{*}{Pal} & \multicolumn{2}{c|} {GoR} & \multirow{2}{*}{Tar} & \multicolumn{2}{c|}  {ZDO} & \multicolumn{2}{c|}  {ZDO-Bits} \\ \cline{3-4} \cline{6-7} \cline{8-9}
 &  & aux	& main &   &  aux	& main &  aux	& main \\ \hline \hline

8192  & 210.170 & 2	& 1 & 36.386  & \cellcolor{lightgray}50.265	& \cellcolor{lightgray}19.200 &\cellcolor{lightgray}43.225 &\cellcolor{lightgray}7.476 \\ \hline

16384 & 310.362 & 2	& 1 & 51.602 & \cellcolor{lightgray}71.197	& \cellcolor{lightgray}27.912 &\cellcolor{lightgray}62.036 &\cellcolor{lightgray}10.174 \\ \hline

32768 & 452.342 & 2	& 1 & 74.074 & \cellcolor{lightgray}100.041	& \cellcolor{lightgray}41.082 &\cellcolor{lightgray}89.889 &\cellcolor{lightgray}13.944 \\ \hline

65536 & 661.663 & 2	& 1 & 103.269 & \cellcolor{lightgray}139.527	& \cellcolor{lightgray}58.476 &\cellcolor{lightgray}126.346 &\cellcolor{lightgray}18.895 \\ \hline
 
131072 & 958.853 & 2	& 1 & 146.923 & \cellcolor{lightgray}197.771	& \cellcolor{lightgray}85.004 &\cellcolor{lightgray}180.624 &\cellcolor{lightgray}26.093 \\ \hline

\end{tabular}}
\end{table}
\begin{figure}[!ht]
\centerline{\includegraphics[width=0.9\textwidth]{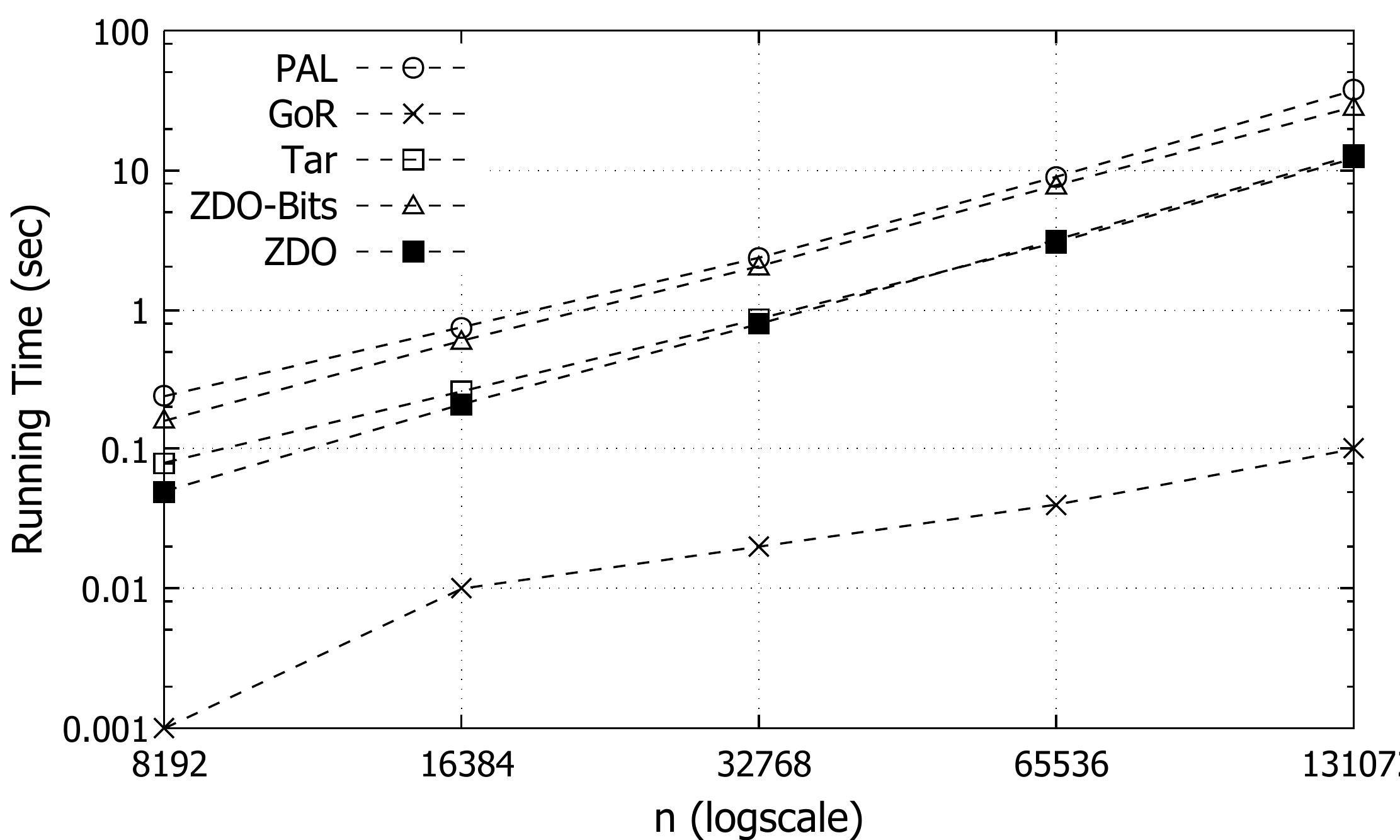}}
\caption{Running time for different algorithms on FN-acyc family}
\label{f13}
\end{figure}
%
\clearpage
\paragraph{Positive to negative arcs.}
This family studies the effect of the percentage of negative arcs on the performance. 
In the positive-to-negative acyclic (P2N-acyc) family, the problem size is fixed $n=16384, m=262144$ and the values of $L$ and $U$ determine the expected fraction {\it f} of negative arcs. The path arc lengths are also selected at random
as for the other arcs.
Table \ref{t12} presents the number of relaxation checks per arc for different algorithms on this family.  
The number of main checks per arc for GoR drops to $1$ and the number of auxiliary checks drops to $2$ when the graph becomes fully negative ($f=100\%$).
Figure \ref{f14} presents the running time for different algorithms on this family. 
We note that the performance of the algorithms, except for GoR, start getting much worse when {\it f} exceeds $40\%$. 

\begin{table}[!htb]
\centering
\caption{\\ Number of relaxation checks per arc for different algorithms on P2N-acyc family ($n=$16384, $m= $262144)\label{t12}}{
\begin{tabular}{|c||c|c|c|c|c|c|c|c|} \hline 
\multirow{2}{*}{$f\%$}  & \multirow{2}{*}{Pal} & \multicolumn{2}{c|} {GoR} & \multirow{2}{*}{Tar} & \multicolumn{2}{c|}  {ZDO} & \multicolumn{2}{c|}  {ZDO-Bits} \\ \cline{3-4} \cline{6-7} \cline{8-9}
 &  & aux	& main &   &  aux	& main &  aux	& main \\ \hline \hline

0  & 1.631 & 2.793	 & 1.245 & 1.514  & \cellcolor{lightgray}1.645	& \cellcolor{lightgray}1.223 &\cellcolor{lightgray}2.049 &\cellcolor{lightgray}0.292 \\ \hline

10  & 2.112 & 3.077	 & 1.379 & 1.845 & \cellcolor{lightgray}2.298	& \cellcolor{lightgray}1.463 &\cellcolor{lightgray}2.470 &\cellcolor{lightgray}0.345 \\ \hline

20  & 6.106 & 5.229	 & 2.394 & 3.857 & \cellcolor{lightgray}6.586	& \cellcolor{lightgray}2.987 &\cellcolor{lightgray}5.105 &\cellcolor{lightgray}0.684 \\ \hline

30  & 22.864 & 10.426	 & 4.975 & 10.093 & \cellcolor{lightgray}17.803	& \cellcolor{lightgray}7.201 &\cellcolor{lightgray}13.906 &\cellcolor{lightgray}1.860 \\ \hline

40  & 91.084 & 20.597	 & 10.098 & 31.964  & \cellcolor{lightgray}50.506	& \cellcolor{lightgray}19.984 &\cellcolor{lightgray}43.941 &\cellcolor{lightgray}6.074  \\ \hline

50  & 291.489 & 29.030	 & 14.367 & 113.871 & \cellcolor{lightgray}173.114	& \cellcolor{lightgray}60.859 &\cellcolor{lightgray}153.359 &\cellcolor{lightgray}24.620 \\ \hline
							
60  & 446.865 & 32.466	 & 16.121 & 270.853 & \cellcolor{lightgray}418.653	& \cellcolor{lightgray}127.511 &\cellcolor{lightgray}345.463 &\cellcolor{lightgray}64.765 \\ \hline

... & .. & .. &.. & .. & \cellcolor{lightgray}.. & \cellcolor{lightgray}.. &\cellcolor{lightgray}.. &\cellcolor{lightgray}.. \\ \hline
 
100  & 472.062 & 2.000	 & 1.000 & 529.338 & \cellcolor{lightgray}819.731	& \cellcolor{lightgray}196.621 &\cellcolor{lightgray}550.907 &\cellcolor{lightgray}143.408 \\ \hline
 
\end{tabular}}
\end{table}
%
\begin{figure}[!ht]
\centerline{\includegraphics[width=0.9\textwidth]{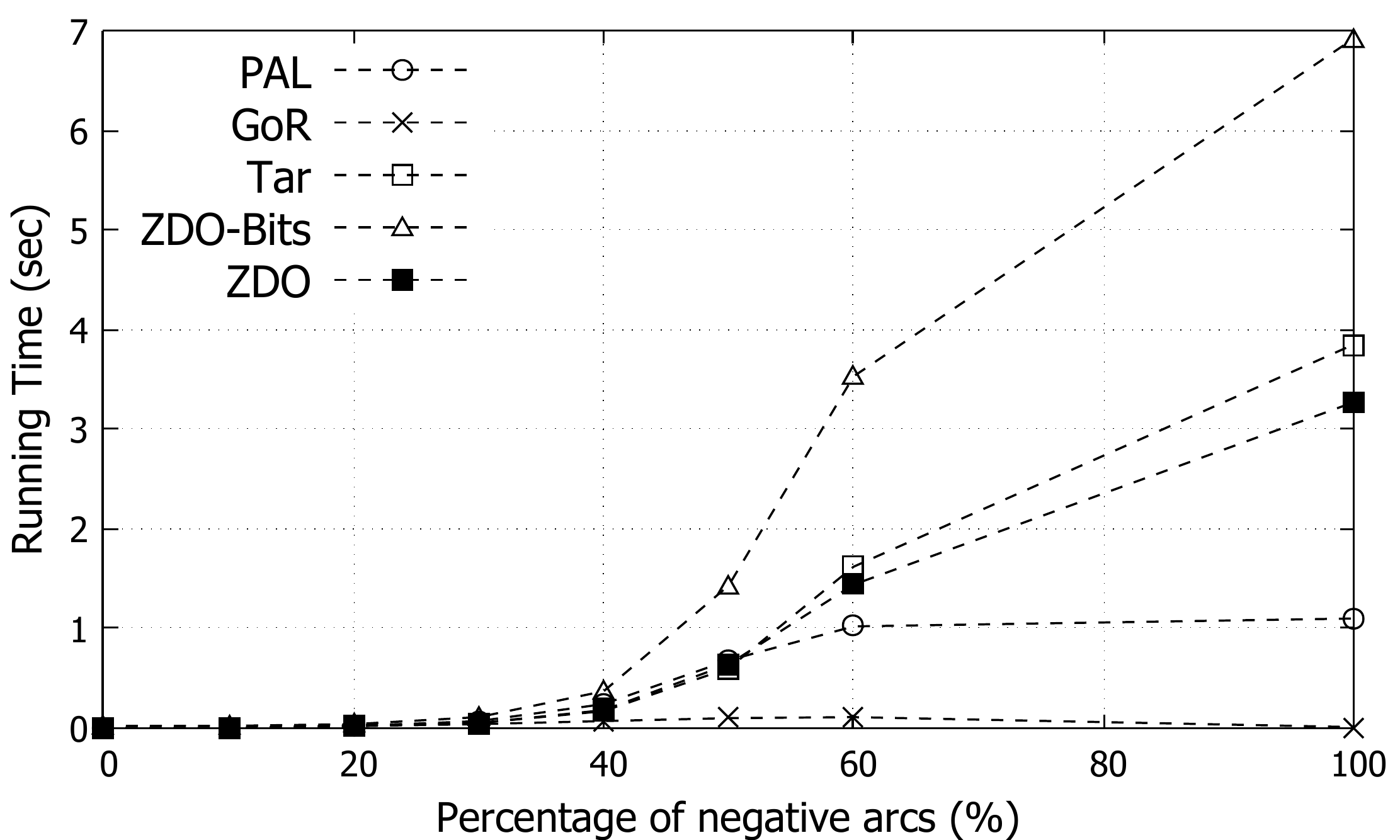}}
\caption{Effect of negative arcs on P2N-acyc family, n= 16384 m = 262144}
\label{f14}
\end{figure}
\clearpage
\subsection{Negative-cycle detection}
This subsection presents the performance of different algorithms in detecting negative cycles. 

\paragraph{RAND05.}
This family is a fixed network with $n = 2 \cdot 10^6$ and $m = 10^7$. The maximum arc length $U$ is fixed at $32000$ and the minimum arc length $L$ varies from $0$ to $-64000$. The running times of different algorithms keep increasing as $L$ decreases until negative cycles start appearing in the graph. Afterward, as $L$ decreases, more short negative cycles appear and their existence can be discovered faster. It is clear that short negative cycles start appearing in $G_d$ before $G_p$, and as GoR traverses $G_d$ so its chance to promptly discover negative cycles with very few relaxation checks is more than others. This is illustrated in Table \ref{t145} and Figure \ref{f145}. 

\begin{table}[!htb]
\centering
\caption{\\ Number of relaxation checks per arc for different algorithms on RAND05 \label{t145}}{
\begin{tabular}{|c||c|c|c|c|c|c|c|c|} \hline 
\multirow{2}{*}{$- L * 10^3$}  & \multirow{2}{*}{Pal} & \multicolumn{2}{c|} {GoR} & \multirow{2}{*}{Tar} & \multicolumn{2}{c|}  {ZDO} & \multicolumn{2}{c|}  {ZDO-Bits} \\ \cline{3-4} \cline{6-7} \cline{8-9}
 &  & aux	& main &   &  aux	& main &  aux	& main \\ \hline \hline

0  & 2.601 & 4.281 & 1.946	    & 2.285       & \cellcolor{lightgray}2.116 & \cellcolor{lightgray}1.830 &\cellcolor{lightgray}2.812 &\cellcolor{lightgray}0.802 \\ \hline

1  & 3.333 & 4.925 & 2.251	    & 2.743       & \cellcolor{lightgray}2.546 & \cellcolor{lightgray}2.200 &\cellcolor{lightgray}3.378 &\cellcolor{lightgray}0.938 \\ \hline

2  & 4.552 & 5.859 & 2.687	   & 3.397        & \cellcolor{lightgray}3.174 & \cellcolor{lightgray}2.757 &\cellcolor{lightgray}4.190 &\cellcolor{lightgray}1.126 \\ \hline
 
4  & 2.973 & 2.030 & 0.525	  & 1.430          & \cellcolor{lightgray}1.131	& \cellcolor{lightgray}0.849 &\cellcolor{lightgray}1.994 &\cellcolor{lightgray}0.592 \\ \hline

8  & 0.685 & 0.053 & 0.003	     & 0.849             & \cellcolor{lightgray}0.878	& \cellcolor{lightgray}0.619 &\cellcolor{lightgray}1.734 &\cellcolor{lightgray}0.519 \\ \hline

16  & 0.138 & 0.001 & 0.000	    & 0.125             & \cellcolor{lightgray}0.116	& \cellcolor{lightgray}0.086 &\cellcolor{lightgray}0.349 &\cellcolor{lightgray}0.083 \\ \hline

32  & 0.016 & 0.000 & 0.000	    & 0.044             & \cellcolor{lightgray}0.050	& \cellcolor{lightgray}0.038 &\cellcolor{lightgray}0.169 &\cellcolor{lightgray}0.038 \\ \hline

64  & 0.025 & 0.000 & 0.000	    & 0.020             & \cellcolor{lightgray}0.023	& \cellcolor{lightgray}0.018 &\cellcolor{lightgray}0.088 &\cellcolor{lightgray}0.018 \\ \hline

\end{tabular}}
\end{table}

\begin{figure}[!htb]
\centerline{\includegraphics[width=0.9\textwidth]{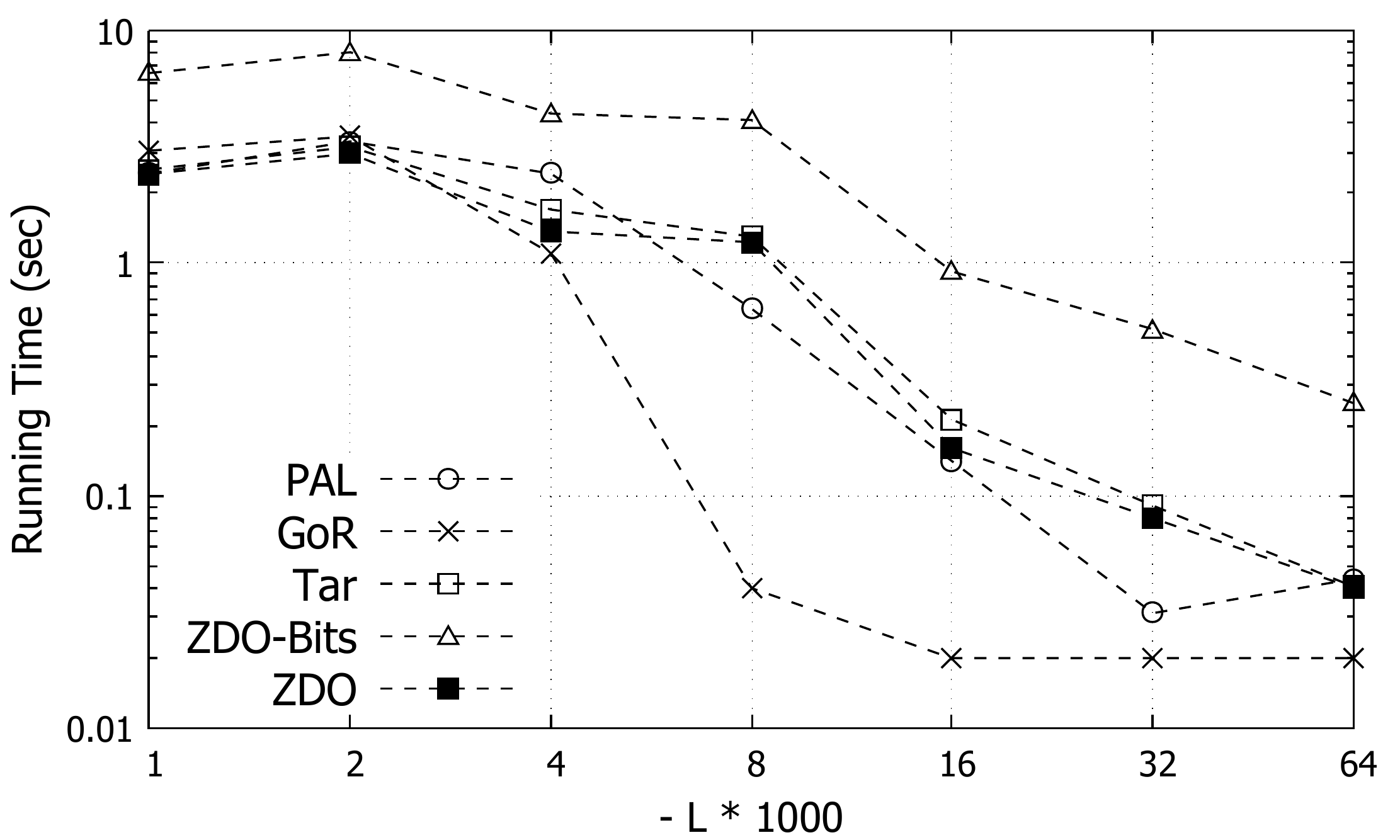}}
\caption{Negative-cycle detection results on Rand05}
\label{f145}
\end{figure}

\clearpage
\paragraph{SQNC05.}
This family is a layered square grid with a Hamilton negative cycle.  
Layer arc lengths are chosen uniformly at random from the interval [1000, 10000]. 
Inter-layer arc lengths, including those from the source, are chosen uniformly at random from the interval [1, 100]. 
We used Tor generator \cite{cherkassky1999negative} to generate this family.
It is reported that Tar is the best algorithm in detecting negative cycles specially for the networks that have long negative cycles 
\cite{cherkassky1999negative}.  
As indicated by Table \ref{t150} and Figure \ref{f150}, to detect the negative Hamiltonian cycle, ZDO-Bits and ZDO outperform Tar and other algorithms with respect to the number of main relaxation checks and the running time respectively. 

\begin{table}[!htb]
\centering
\caption{\\ Number of relaxation checks per arc for different algorithms on SQNC05 \label{t150}}{
\begin{tabular}{|c||c|c|c|c|c|c|c|c|} \hline 
\multirow{2}{*}{n}  & \multirow{2}{*}{Pal} & \multicolumn{2}{c|} {GoR} & \multirow{2}{*}{Tar} & \multicolumn{2}{c|}  {ZDO} & \multicolumn{2}{c|}  {ZDO-Bits} \\ \cline{3-4} \cline{6-7} \cline{8-9}
 &  & aux	& main &   &  aux	& main &  aux	& main \\ \hline \hline

64 & 18.996 & 21.524	& 9.686 & 10.109  & \cellcolor{lightgray}9.943	& \cellcolor{lightgray}9.545 &\cellcolor{lightgray}10.921 &\cellcolor{lightgray}3.992 \\ \hline

128 & 23.154 & 24.478	& 11.141 & 12.000 & \cellcolor{lightgray}11.778	& \cellcolor{lightgray}11.385 &\cellcolor{lightgray}12.863 &\cellcolor{lightgray}4.653 \\ \hline

256 & 28.234 & 28.322	& 13.048 & 14.154 & \cellcolor{lightgray}13.914	& \cellcolor{lightgray}13.512 &\cellcolor{lightgray}15.094 &\cellcolor{lightgray}5.416 \\ \hline

512 & 34.071 & 33.011	& 15.384 & 16.383 & \cellcolor{lightgray}16.119	& \cellcolor{lightgray}15.709 &\cellcolor{lightgray}17.391 &\cellcolor{lightgray}6.196 \\ \hline
 
1024& 37.101  & 35.130	& 16.423 & 17.864 & \cellcolor{lightgray}17.581	& \cellcolor{lightgray}17.159 &\cellcolor{lightgray}18.942 &\cellcolor{lightgray}6.731 \\ \hline

\end{tabular}}
\end{table}

\begin{figure}[!ht]
\centerline{\includegraphics[width=0.9\textwidth]{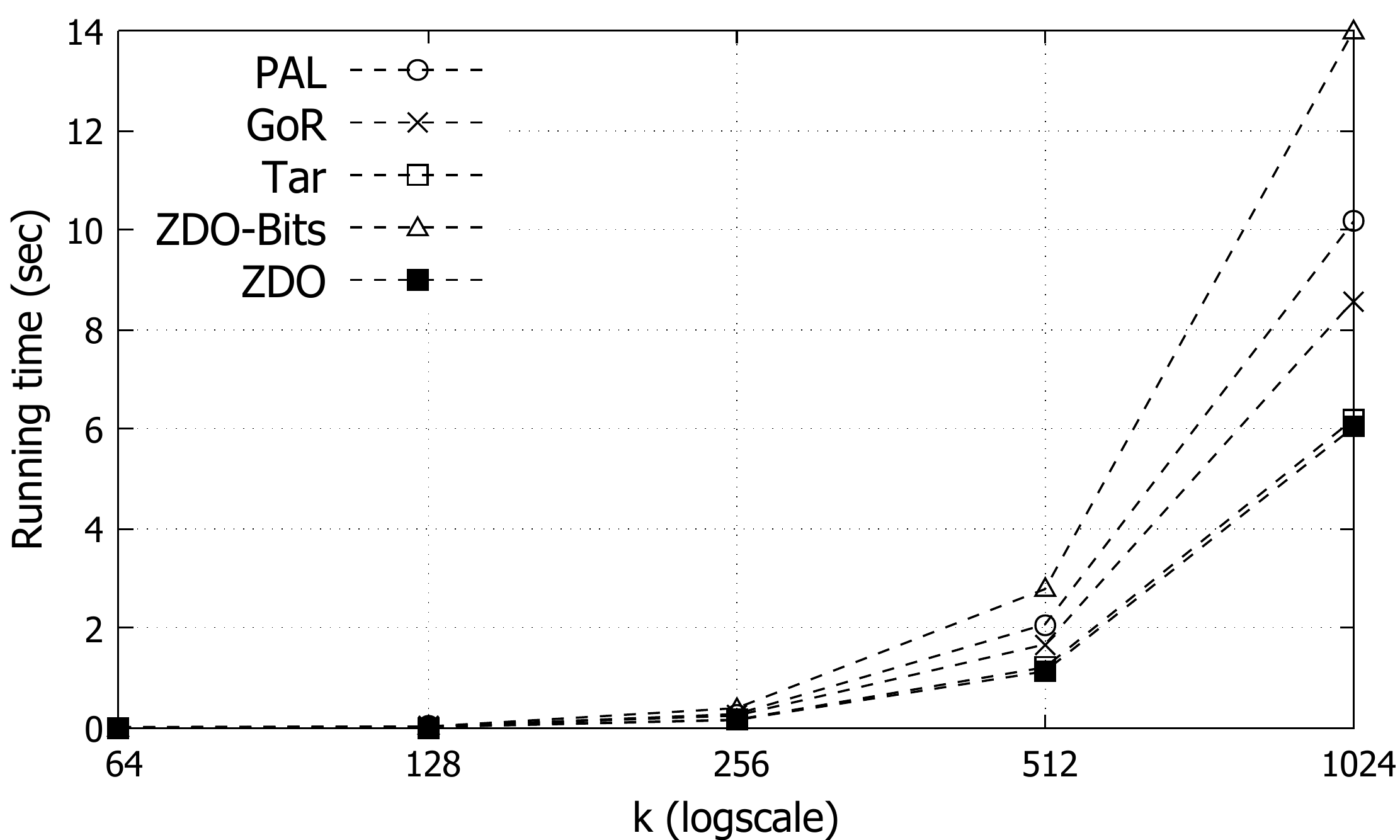}}
\caption{Negative cycle detection results in SQNC05}
\label{f150}
\end{figure}
\clearpage
\section{Summary and comments on the results}

This section summarizes our observations as follows.

\begin{itemize}

\item ZDO achieves the best running time on eleven out of fourteen well-known families. 
It comes second after GoR on two families, and third after GoR and Pal on one family.

\item ZDO is consistently superior over Tar, the most efficient state-of-the-art algorithm, on all the families.

\item GoR is the fastest algorithm on the FN-acyc and P2N-acyc families, 
and is generally very efficient for acyclic graphs as it runs in linear time. 
GoR is also the champion and is slightly faster than ZDO on the NH-grid family.  

\item Pal achieves a comparable performance with ZDO and Tar in FP-acyc family. Also, it has the best performance after GoR in P2N-acyc family. In negative cycle detection, Pal can detect the short negative cycles in RAND05 family as fast as ZDO.

\item Both ZDO and ZDO-Bits are superior on the Bad-GoR family. ZDO-bits is more than three orders of magnitude faster than the other algorithms. This illustrates that addressing the relaxable graph $G'_d$ by our technique would possibly improve the performance significantly when compared to addressing the admissible graph $G_d$ by other algorithms.

\item ZDO can detect long negative cycles on SQNC05 family faster than the other algorithms. 
However, for short negative cycles on RAND05 family, ZDO is the second best after GoR.

\end{itemize}

\noindent The following comments justify the performance of our algorithm.

\begin{itemize}

\item The main advantage of our algorithm is that it deals with and scans a fewer subset of the vertices in each round compared to the other algorithms. The subtree-disassembly heuristic, in a sense, does that by excluding some vertices so as not to be promptly scanned.

\item Even though in several cases the ZDO implementation performs a number of relaxation checks comparable to that of Tar,
the new proposed algorithms are still faster. One of the reasons is that Tar uses the subtree-disassembly heuristic to decrease the number of candidate vertices
by removing them from the queue after being inserted, while for ZDO most of these vertices are not inserted in the queue in the first place.

\item ZDO-Bits is faster than ZDO for the Star family as it encodes a huge number of checks in the bit vectors. Also,
it is the second fastest algorithm after GoR in solving the Star family. If there are many positive arcs on the chain, the problem becomes even hard for GoR. ZDO-Bits is more than three orders of magnitude faster than the other algorithms.

\item It is obvious that the cost of a main relaxation check is significantly more than the cost of an auxiliary check. This is due to the overhead of the potential update and parent update in addition to the call for the subtree-disassembly procedure in the case of the main relaxation check.

\item Although ZDO-Bits always performs the fewest number of main relaxation checks, it loses the competition when it comes to the running time because of the high cost of setting and clearing the bits. 
When we needed a large number of bits (e.g. in dense graphs), we split the bits among multiple words, each with 64 bits (max bit word in our machine). This increased the overhead for fetching the designated word before setting or clearing the bits.
 
\end{itemize}

\chapter{Conclusion}
\label{conclusions}
In this thesis, a review of the single source shortest path problem and the negative cycle detection problem is proposed.
In addition, a new algorithm for the shortest-path problem and for detecting negative cycles is discussed.
The main idea behind the proposed algorithm is to simultaneously consider both the relaxable graph and the parent graph and select the fewest most-effective vertices to scan.
The proposed algorithm achieves the same $O(n \cdot m)$ time bound as the traditional Bellman-Ford-Moore algorithm but outperforms it and other state-of-the-art algorithms in practice.

In addition to the running time, the results were presented using a new performance metric that is more realistic than the number of scans. 
Experiments show that the proposed algorithm outperforms the state-of-the-art algorithms in practice while maintaining the same $O(n \cdot m)$ time bound as the other algorithms.


\addtocontents{toc}{\vspace{2em}} 

\appendix 

\addtocontents{toc}{\vspace{2em}}  
\backmatter

\label{Bibliography}

\bibliographystyle{plain}

\end{document}